\documentclass[5p]{elsarticle}
\usepackage{lineno,hyperref}
\modulolinenumbers[5]

\usepackage{algorithmic}

\usepackage{amsmath}
\newtheorem{definition}{Definition}
\newtheorem{theorem}{Theorem}
\newenvironment{proof}[1][Proof]{\noindent\textbf{#1.} }{\ \rule{0.5em}{0.5em}}

\usepackage{color}
\usepackage{url}
\usepackage{subfig}
\usepackage{textcase}

\usepackage{amsmath,amssymb,amsfonts,xspace}
\DeclareMathOperator*{\argmax}{arg\,max}
\newcommand{\lbl}{\mathtt{lbl}\xspace}
\newcommand{\tail}{\mathtt{tl}\xspace}
\newcommand{\head}{\mathtt{hd}\xspace}
\newcommand{\dist}{\mathtt{dist}\xspace}
\newcommand{\diam}{\mathtt{diam}\xspace}
\newcommand{\nb}{\mathtt{N}\xspace}
\newcommand{\pr}{\mathtt{priority}\xspace}
\newcommand{\est}{\mathtt{est}\xspace}
\newcommand{\degree}{\mathtt{deg}\xspace}
\newcommand{\hops}{\mathtt{hp}\xspace}
\newcommand{\rad}{\mathtt{rad}\xspace}
\newcommand{\subs}{\mathtt{subs}\xspace}

\journal{JoWS}

\begin{document}

\begin{frontmatter}

\title{Relaxing Relationship Queries on Graph Data}

\author[nju]{Shuxin~Li}
\ead{sxli@smail.nju.edu.cn}

\author[nju]{Gong~Cheng\corref{cor}}
\ead{gcheng@nju.edu.cn}
\cortext[cor]{Corresponding author; tel: +86~(0)25~89680923; fax: +86~(0)25~89680923}

\author[uta]{Chengkai~Li}
\ead{cli@uta.edu}


\address[nju]{National Key Laboratory for Novel Software Technology, Nanjing University, Nanjing 210023, China}
\address[uta]{Department of Computer Science and Engineering, University of Texas at Arlington, Arlington, Texas, United States}

\begin{abstract}
In many domains we have witnessed the need to
search a large entity-relation graph for direct and indirect relationships between a set of entities specified in a query.
A search result, called a semantic association (SA),
is typically a compact (e.g., diameter-constrained) connected subgraph containing all the query entities.
For this problem of SA search,
efficient algorithms exist but will return empty results
if some query entities are distant in the graph.
To reduce the occurrence of failing query and provide alternative results,
we study the problem of query relaxation in the context of SA search.
Simply relaxing the compactness constraint will sacrifice the compactness of an SA,
and more importantly, may lead to performance issues and be impracticable.
Instead, we focus on removing the smallest number of entities from the original failing query,
to form a maximum successful sub-query
which minimizes the loss of result quality caused by relaxation.
We prove that
verifying the success of a sub-query turns into
finding an entity (called a certificate) that satisfies a distance-based condition about the query entities.
To efficiently find a certificate of the success of a maximum sub-query,
we propose a best-first search algorithm
that leverages distance-based estimation to effectively prune the search space.
We further improve its performance by adding two fine-grained heuristics:
one based on degree and the other based on distance.
Extensive experiments over popular RDF datasets
demonstrate the efficiency of our algorithm, which is more scalable than baselines.
\end{abstract}
\begin{keyword}
semantic association search, complex relationship, query relaxation, graph data
\end{keyword}

\end{frontmatter}


\section{Introduction}\label{sect:intro}
Graph data (e.g., RDF data) representing binary relations between entities
is becoming the back end of increasingly many applications.
Graphs are particularly suitable for answering \emph{relationship queries}.
As a simple example, with the academic graph in Fig.~\ref{fig:erg},
answering a query like \emph{how is Dan related to ISWC 2019}
could be to conveniently look up arcs
that connect two particular vertices representing the two entities mentioned in the query:
\texttt{Dan} and \texttt{ISWC2019},
which are called \emph{query entities}.
In the literature, a relationship between two or more query entities
is commonly referred to as a \emph{semantic association}~(SA).
For two query entities,
an SA is usually a path or a path-like subgraph that connects them~\cite{rho,explass,vldb04,brahms,conkar,semrank,rex,explain,relfinder}.
More generally, for a set of two or more query entities,
an SA is a compact subgraph that connects all the query entities~\cite{sisp,iswc16,tkde17,ming,star,ceps}.
For example, the SA shown on the right-hand side of Fig.~\ref{fig:ppl},
which is a subgraph of the academic graph in Fig.~\ref{fig:erg},
is an answer to \emph{how are Alice, Bob, and Dan connected}.
It shows that Alice and Bob have papers accepted at a conference
which Dan is a PC member of.

\begin{figure}[!t]
\centering
\includegraphics[width=0.9\columnwidth]{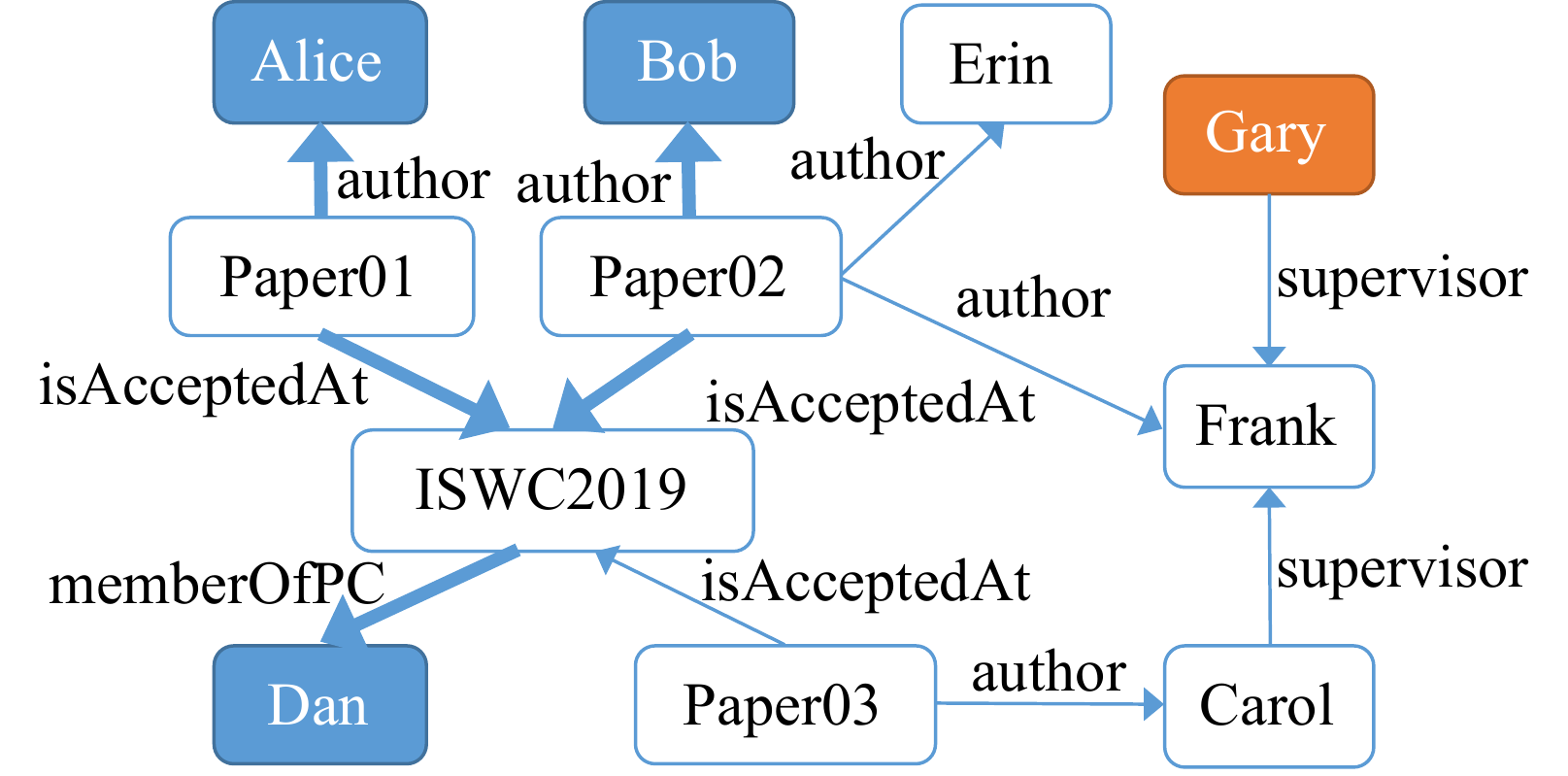}
\caption{An example of entity-relation graph,
where the five bold arcs induce a connected subgraph
representing a relationship between three entities in a relationship query:
\texttt{Alice}, \texttt{Bob}, and \texttt{Dan}.
The subgraph is referred to as an SA which is a result of that query.}\label{fig:erg}
\end{figure}

\begin{figure*}[!t]
\centering
\includegraphics[width=\linewidth]{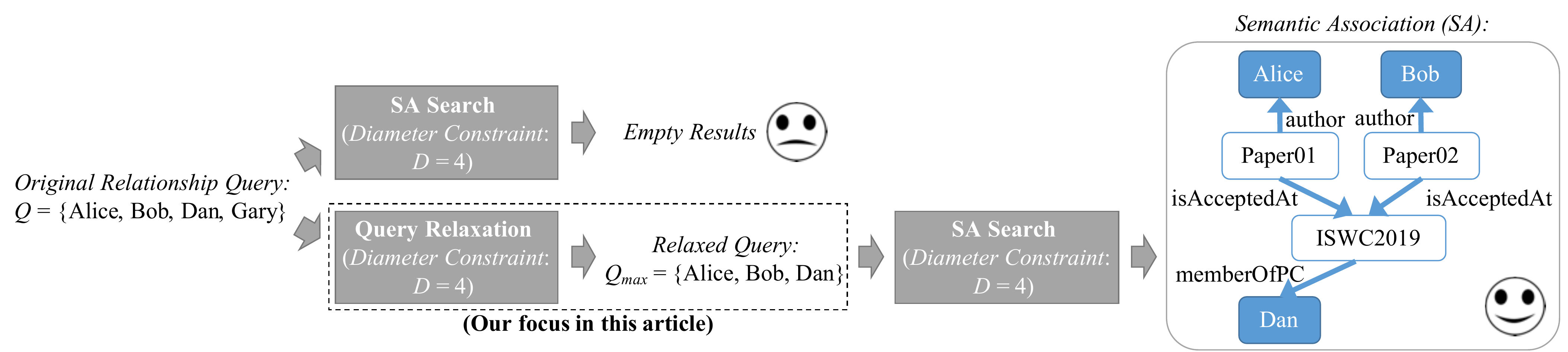}
\caption{An example of query relaxation for SA search on the entity-relation graph in Fig.~\ref{fig:erg}.}\label{fig:ppl}
\end{figure*}

\textbf{Application of SA Search.}
Searching a graph for SAs that connect a set of query entities,
called \emph{SA search} for short,
is a well-established research problem
and has found application in many domains.
For example, given a graph about the national security domain
including organizations, countries, people, terrorists, terrorist acts, etc.,
it is useful to detect notable SAs that connect a group of suspect airline passengers in a given flight~\cite{appns}.
In a social network, e.g., a co-authorship network of scientists,
SAs can clearly depict the key relationships between a group of scientists,
without the disturbance of many remote and uncorrelated scientists~\cite{ceps,sisp}.
SA search is also helpful in biomedical research~\cite{appbio},
and is an underlying technique for keyword search on graphs~\cite{banks,dpbf,blinks,banks2,rclique,le,ease,pruneddp}.

\textbf{Limitations of Existing Solutions.}
SA search may \emph{fail and yield empty results},
which can disappoint users of an SA search system.
It happens when some query entities are disconnected in the graph
and hence there are no connected subgraphs containing all the query entities.
More generally, it also happens when some query entities are connected but are far away from each other in the graph,
so that existing solutions that aim at structurally compact SAs may fail to produce any results.
Such solutions adopt various kinds of \emph{compactness constraints} on allowable SAs,
mainly in order to bound the search space and achieve satisfactory performance.
They place an upper bound on
the length of a path-structured SA~\cite{explass,brahms,explain},
on the number of vertices in an SA~\cite{sisp,rex,ming,ceps},
or on the radius or diameter of an SA~\cite{iswc16,tkde17,ease}.
For example, in Fig.~\ref{fig:erg},
if the diameter of an SA is bounded by~4,
no SAs in the graph can connect all of \texttt{Alice}, \texttt{Bob}, \texttt{Dan}, and \texttt{Gary},
because the distance between \texttt{Alice} and \texttt{Gary} is~5 which exceeds the bound,
even though \texttt{Alice}, \texttt{Bob}, and \texttt{Dan} are sufficiently close to each other.

Bounding is necessary and desired for practical applications where top-ranked SAs are requested. Computing the top-ranked SA in an unbounded space is usually formulated as a Steiner tree problem~\cite{star}, and is NP-hard. The performance of existing solutions is acceptable only when the search space is bounded and the bound is set to very small values~\cite{emp}, because the number of candidate SAs is exponential to the size of an allowable SA.
On the other hand, solutions with compactness constraints are believed cost-effective
because the relationship represented by a large SA,
which would be disregarded by those solutions,
is usually not meaningful or interesting to users~\cite{ease},
as demonstrated by our recent user study~\cite{tkde17}.

In spite of the necessity and benefits of adopting compactness constraint,
\emph{query failure would occur more often},
which has not been addressed in previous research.

\textbf{Research Goal.}
Subject to a predetermined compactness constraint which is practically needed for SA search,
our goal in this article is to reduce the occurrence of failing query and improve the usability of an SA search system,
Consequently, a better trade-off will be established between the performance of search and the quality of search results.
To this end, we will study \emph{query relaxation} techniques for SA search that can provide alternative search results instead of empty results.

\textbf{Methods and Challenges.}
One straightforward approach is to \emph{relax the compactness constraint} (e.g., lifting the upper bound of diameter to~5 in the above example),
which can produce an SA for the above query.
Despite the sacrifice of compactness,
this approach may not fundamentally solve the problem.
Recall that compactness constraint is introduced mainly in order to bound the search space.
It cannot be arbitrarily relaxed as needed
but has to be maintained at a small value
in order to achieve acceptable performance of search.
Our experiment results in Section~\ref{sect:expq} will show that: considerably many queries still fail
even though the compactness constraint has been lifted to the largest value to which existing search algorithms can scale.

Another approach, which we will explore in this article,
is to find SAs that connect not all but part of the query entities,
i.e., to \emph{relax the query entities}.
As illustrated in Fig.~\ref{fig:ppl},
without lifting the upper bound of diameter,
an SA of diameter~4 that connects \texttt{Alice}, \texttt{Bob}, and \texttt{Dan} (but not \texttt{Gary}) can be found as a search result and returned to the user.
This kind of query relaxation technique has been familiar to ordinary users.
For example, as a comparable application, Google Search retrieves documents that do not contain all the query keywords,
and it explicitly indicates the missing keywords for each document in the search results pages.
Following this paradigm,
research challenges in the context of SA search include:
\begin{enumerate}
    \item how to optimize the quality of relaxed search results (i.e., SAs) in terms of result completeness, and
    \item how to efficiently perform the optimization.
\end{enumerate}

\textbf{Research Contributions.}
Relaxation techniques have been considered for
relational queries~\cite{qrdb},
graph queries~\cite{qrsparql1},
XML search~\cite{qrxml},
and entity search~\cite{qrentity}.
However, query relaxation for SA search is a new research problem.
The task here is to remove some entities from a failing relationship query
such that the remaining sub-query could be successful,
i.e., an SA subject to the compactness constraint can be found as a result of the sub-query.
To optimize the quality of relaxed SAs and meet the first challenge,
we aim to minimize the loss of search result completeness by \emph{removing the smallest number of query entities},
so that the largest proportion of query entities can be connected by relaxed SAs.
This problem is non-trivial in consideration of
the magnitude of graph data and the allowable response time of a search system.
To develop an efficient solution and meet the second challenge,
our technical contributions are summarized as follows.
\begin{itemize}
  \item We prove that: the success of a relationship query can be indirectly but more efficiently verified by finding an entity (called a certificate) that satisfies a distance-based condition about the query entities.
  Following that, we devise a polynomial-time algorithm for query relaxation called CertQR, which is more scalable than an intuitive exponential-time solution adapted from a state-of-the-art algorithm for SA search~\cite{iswc16}.
  \item We then devise a best-first search algorithm called CertQR+, which leverages distance-based estimation to effectively prune the search space of CertQR.
  Furthermore, to guide CertQR+ to find an optimum solution earlier, we introduce two fine-grained heuristics: one based on degree and the other based on distance.
  The combination of these techniques achieves considerable performance improvement.
\end{itemize}

The remainder of this article is organized as follows.
Section~\ref{sect:probl} formulates the problem.
Section~\ref{sect:basic} and Section~\ref{sect:opt} present CertQR and CertQR+, respectively.
Section~\ref{sect:eval} reports experiments.
Section~\ref{sect:rw} discusses related work.
Section~\ref{sect:concl} concludes the article with future work.
\section{Problem}\label{sect:probl}
\subsection{Preliminaries}
The terminology defined here is compatible with that used in our previous work~\cite{iswc16,tkde17}.
Entities and their binary relations form a graph.
\begin{definition}[entity-relation graph]
  An entity-relation graph is a finite directed labeled graph denoted by $G = \langle E_G,A_G,\mathbb{R},\lbl_G \rangle$ where
  \begin{itemize}
    \item $E_G$~is a finite set of entities as vertices,
    \item $A_G$~is a finite set of arcs,
    each arc $a \in A_G$ directed from its tail vertex $\tail(a) \in E_G$ to its head vertex $\head(a) \in E_G$, and
    \item $\lbl_G : A_G \mapsto \mathbb{R}$ is a function that
    labels each arc $a \in A_G$ with a binary relation $\lbl_G(a) \in \mathbb{R}$.
  \end{itemize}
\end{definition}
\noindent Figure~\ref{fig:erg} shows an entity-relation graph, which will be used as a running example in this article.
RDF graph is a kind of entity-relation graph. Following~\cite{iswc16,tkde17}, we focus on the relations between instance-level entities and hence we ignore \texttt{rdf:type} and literals.

Although $G$~is directed, it is often treated as an undirected graph in SA search.
For example, when referring to a path in~$G$, we do not require its arcs to all go the same direction.
When referring to a tree in~$G$, we do not require its arcs to all go towards or all go away from a particular vertex.

A user may intend to search an entity-relation graph
for relationships between a particular set of entities.
\begin{definition}[relationship query]
  A relationship query, or a query for short,
  consists of $n$~entities ($n \geq 2$) in~$G$ denoted by $Q = \{qe_1, \ldots, qe_n\}$,
  where $qe_1, \ldots, qe_n \in E_G$ are called query entities.
\end{definition}
\noindent For example, $\{\texttt{Alice}, \texttt{Bob}, \texttt{Dan}\}$
is a relationship query submitted to the entity-relation graph in Fig.~\ref{fig:erg}.

An SA as a result of a relationship query is a minimal subgraph connecting all the query entities.
It is actually a Steiner tree that spans all the query entities.
We follow~\cite{iswc16,tkde17} to put an upper bound on the diameter of an allowable SA as a compactness constraint.
Before giving the formal definition of SA, we review some graph terminology.
\begin{definition}[graph terminology]
  The length of a path is the number of arcs it uses.
  The distance between two vertices, denoted by~$\dist(\cdot,\cdot)$, is the length of a shortest path between them,
  or~$+\infty$ if no such path exists.
  The diameter of a graph $G = \langle E_G,A_G,\mathbb{R},\lbl_G \rangle$, denoted by $\diam(G)$, is the largest distance between pairs of vertices:
  \begin{equation}
      \diam(G) = \max_{e_i,e_j \in E_G}{\dist(e_i,e_j)} \,.
  \end{equation}
  The eccentricity of a vertex is the largest distance between this vertex and other vertices in the graph. The radius of a graph, denoted by $\rad(\cdot)$, is the minimum eccentricity in the graph:
  \begin{equation}
      \rad(G) = \min_{e_i \in E_G}{\max_{e_j \in E_G}{\dist(e_i,e_j)}} \,.
  \end{equation}
  A central vertex is a vertex of minimum eccentricity.
  The neighbors of a vertex~$e$, denoted by $\nb(e)$, are the vertices that are adjacent from/to it:
  \begin{equation}
      \nb(e) = \{ e' \in E_G : \exists a \in A_G, \{\tail(a),\head(a)\}=\{e,e'\} \} \,.
  \end{equation}
\end{definition}

Now we formally define SA.
\begin{definition}[semantic association]\label{def:sa}
  A result of a relationship query~$Q$ is called a semantic association~(SA), denoted by $x = \langle E_x,A_x \rangle$,
  which is a subgraph of~$G$ with vertices~$E_x \subseteq E_G$ and arcs~$A_x \subseteq A_G$ satisfying
  \begin{itemize}
    \item $x$~contains all the query entities in~$Q$, i.e., $Q \subseteq E_x$,
    \item $x$~is connected,
    \item $x$~is minimal, i.e., none of its proper subgraphs satisfy the above two conditions, and
    \item $\diam(x) \leq D$, where $D$~is a diameter constraint.
  \end{itemize}
\end{definition}
\noindent From the minimality of~$x$, we infer that
$x$~is a (Steiner) tree where all the leaf vertices are query entities.
For example, given $D=4$,
the tree shown in Fig.~\ref{fig:ppl} is an SA of diameter~4 for the query $\{\texttt{Alice}, \texttt{Bob}, \texttt{Dan}\}$.

\subsection{Problem Statement}
SA search is to search an entity-relation graph for SAs
that are results of an input relationship query.
Depending on whether SA search yields empty results,
we classify relationship queries into two types.
\begin{definition}[successful query and failing query]
  A relationship query~$Q$ is successful
  if at least one SA can be found in~$G$ under the predefined diameter constraint.
  Otherwise, $Q$~is called a failing query.
\end{definition}
\noindent For example, the relationship query $\{\texttt{Alice}, \texttt{Bob}, \texttt{Dan}\}$ submitted to the entity-relation graph in Fig.~\ref{fig:erg}
is a successful query under $D \geq 4$
because the subgraph shown in Fig.~\ref{fig:ppl} is an SA of diameter~4 for this query.
The query $\{\texttt{Alice}, \texttt{Bob}, \texttt{Dan}, \texttt{Gary}\}$ is a failing query when $D \leq 4$.

\begin{definition}[sub-query]
  For two relationship queries~$Q$ and~$Q'$ satisfying $Q' \subseteq Q$, $Q'$~is called a sub-query of~$Q$.
\end{definition}
\noindent Let $\subs(Q)$ be the set of all the sub-queries of~$Q$.
For example, the query $\{\texttt{Alice}, \texttt{Bob}, \texttt{Dan}\}$ has four sub-queries:
$\{\texttt{Alice}, \texttt{Bob}\}$, $\{\texttt{Alice}, \texttt{Dan}\}$, $\{\texttt{Bob}, \texttt{Dan}\}$, and itself.
Recall that a relationship query contains at least two entities.

In case a relationship query fails, we aim to relax it by finding a sub-query such that
it is successful and it undergoes the smallest change from the original query.
We measure the change by the number of removed query entities.
\begin{definition}[query relaxation for SA search]\label{def:qr}
  Given a relationship query~$Q$ submitted to an entity-relation graph~$G$ under a diameter constraint~$D$,
  the problem of query relaxation is to find a sub-query of~$Q$, denoted by~$Q_{max}$,
  such that $Q_{max}$~is successful and contains the largest number of query entities in~$Q$:
  \begin{displaymath}
    Q_{max} = \argmax_{Q' \in \subs(Q) \text{ and } Q' \text{ is successful}}{|Q'|} \,.
  \end{displaymath}
  We define $Q_{max} = \emptyset$ if none of the sub-queries of~$Q$ are successful.
  We will have $Q_{max}=Q$ if $Q$~itself is successful.
\end{definition}
\noindent For example, under $D \geq 4$
on the entity-relation graph~$G$ in Fig.~\ref{fig:erg},
the relationship query $Q=\{\texttt{Alice}, \texttt{Bob}, \texttt{Dan}\}$ 
is successful, and hence $Q_{max}=Q$.
However, the query $Q=\{\texttt{Alice}, \texttt{Bob}, \texttt{Dan}, \texttt{Gary}\}$ fails,
and we have $Q_{max}=\{\texttt{Alice}, \texttt{Bob}, \texttt{Dan}\}$ or $\{\texttt{Bob}, \texttt{Dan}, \texttt{Gary}\}$.
Maximum successful sub-queries may not be unique.

The output of our problem is a relaxed query rather than the results of the said query. This separation of query relaxation and SA search decouples our work from downstream tasks, and hence maximizes our research’s applicability in the real world. For example, some downstream approach searches for all possible SAs and mines their frequent patterns ~\cite{iswc16}, whereas other approaches can perform search algorithms that are tailored to directly find top-ranked SAs according to various ranking criteria~\cite{tkde17}. All these downstream solutions can use our approach to preprocess an input relationship query to avoid query failure. By contrast, if we focus on directly returning the results of a relaxed query, we will have to choose a specific downstream task (e.g., a specific ranking function) and design a hybrid approach that combines query relaxation and SA search, which may be more efficient than two separate steps. However, its applicability would be limited. We will explore this direction in future work.

\subsection{Discussion on Intuitive Solutions}\label{sect:int}
We are among the first to study query relaxation for SA search.
Before introducing our algorithms, we discuss the shortcomings of two intuitive solutions.
The discussion could help explain the difficulty of the problem.

The first intuitive solution is to exhaustively check all possible sub-queries
in non-increasing order of the number of query entities they contain,
and either return the first sub-query that is successful
or return~$\emptyset$ if none of the sub-queries are successful.
To check a sub-query,
we could perform an existing algorithm for SA search, e.g.,~\cite{iswc16}.
Immediately when the search algorithm finds the first result,
we know that the sub-query is successful.
Otherwise, if the algorithm ends without outputting any result,
the sub-query fails.
One shortcoming of this solution is its \emph{exponential running time}.
The time of performing SA search to check a sub-query is exponential to the diameter constraint~\cite{iswc16},
and the number of sub-queries to be checked is exponential to the number of entities in the original query.
We will optimize this solution and use it as a baseline in our experiments.
However, more scalable solutions are our focus in this article.

The second intuitive solution is to calculate the distance between every pair of query entities.
The idea is:
in a successful sub-query,
the distance between every pair of query entities is not larger than the diameter constraint.
However, this is \emph{only necessary but not sufficient} for the success of a sub-query.
The issue is related to the minimality of an SA
which is required to be a tree,
in order to follow common practice in the research of SA search
and be compatible with existing search algorithms and ranking criteria, e.g.,~\cite{iswc16,tkde17,star}.
For a counterexample, consider a small entity-relation graph that is a triangle,
where all the three entities form a relationship query.
Although the distance between every pair of query entities is~1,
the query still fails under $D=1$,
because all the spanning trees of the graph have a diameter of~2,
thereby violating the diameter constraint.
\section{\texorpdfstring{C\MakeLowercase{ert}QR}{CertQR}: Certificate-based Algorithm}\label{sect:basic}
We outline our idea as follows.
The success of a relationship query can be indirectly but more efficiently verified by finding an entity (called a certificate) that satisfies a distance-based condition about the query entities.
A maximum sub-query that a particular entity is a certificate of the success of
can be found in polynomial time.
Clearly there are at most linearly many candidate certificates to consider.
Therefore, our proposed algorithm for query relaxation, called CertQR, runs in polynomial time.

We will introduce the notion of certificate in Section~\ref{sect:cert},
and present the CertQR algorithm in Section~\ref{sect:bscalg}.

\subsection{Certificate of Successful Sub-query}\label{sect:cert}
We will informally illustrate our idea, and then formalize the notion of certificate.
We use $\dist$ to specifically denote the distance between two entities in~$G$.

\subsubsection{Basic Idea}
Under a diameter constraint~$D$,
a relationship query~$Q$ submitted to an entity-relation graph~$G$
is successful \emph{only if} there is an entity $c \in E_G$ that
is at most $\left\lceil{\frac{D}{2}}\right\rceil$~hops away from every query entity in~$Q$.
The existence of~$c$ is a necessary condition for the success of~$Q$, which we will prove later.
For example,  under $D=4$,
the query $\{\texttt{Alice}, \texttt{Bob}, \texttt{Dan}, \texttt{Gary}\}$
submitted to~$G$ in Fig.~\ref{fig:erg} is not successful,
because no such entity~$c$ exists.
Its sub-query $\{\texttt{Alice}, \texttt{Bob}, \texttt{Dan}\}$ 
is successful, and we have $c=\texttt{ISWC2019}$ because
\begin{equation}\label{eq:cert}
\begin{split}
  \dist(\texttt{ISWC2019},\texttt{Alice}) & = 2 \leq \left\lceil{\frac{D}{2}}\right\rceil \,, \\
  \dist(\texttt{ISWC2019},\texttt{Bob}) & = 2 \leq \left\lceil{\frac{D}{2}}\right\rceil \,, \\
  \dist(\texttt{ISWC2019},\texttt{Dan}) & = 1 \leq \left\lceil{\frac{D}{2}}\right\rceil \,.
\end{split}
\end{equation}

When $D$~is even, the existence of~$c$ is also a sufficient condition for the success of~$Q$,
which we will prove later.
However, it is not a sufficient condition when $D$~is odd.
For example, when $D=3$, the query $\{\texttt{Alice}, \texttt{Bob}, \texttt{Dan}\}$ is not successful
even though the inequalities in Eq.~(\ref{eq:cert}) hold.
The issue is due to query entities that are exactly $\left\lceil{\frac{D}{2}}\right\rceil$~hops away from~$c$,
which are referred to as \emph{critical query entities},
e.g., \texttt{Alice} and \texttt{Bob} in the above example which satisfy
\begin{equation}\label{eq:d3}
\begin{split}
  \dist(\texttt{Alice},\texttt{Bob}) & = \dist(\texttt{Alice},\texttt{ISWC2019}) \\
  & \quad + \dist(\texttt{ISWC2019},\texttt{Bob}) \\
  & = \left\lceil{\frac{D}{2}}\right\rceil + \left\lceil{\frac{D}{2}}\right\rceil > D \,, \\
\end{split}
\end{equation}
\noindent so the diameter constraint is violated.
To establish a sufficient condition when $D$~is odd, we observe that:
if $c$~has a neighbor $c' \in E_G$ such that
it is $(\left\lceil{\frac{D}{2}}\right\rceil - 1)$~hops away from \emph{every} critical query entity,
the distance between every pair of critical query entities
will not be larger than $2(\left\lceil{\frac{D}{2}}\right\rceil - 1) = D-1$,
thereby complying with the diameter constraint.
Using the existence of~$c'$ as an additional condition, when $D=3$ in the above example,
$c=\texttt{ISWC2019}$ will not mistakenly imply the success of the query $\{\texttt{Alice}, \texttt{Bob}, \texttt{Dan}\}$
because none of the neighbors of~$c$ qualify for~$c'$.
In the meantime, $c=\texttt{ISWC2019}$ correctly implies the success of the query $\{\texttt{Dan}, \texttt{Erin}, \texttt{Frank}\}$
with $c'=\texttt{Paper02}$.
The existence of $\langle c,c' \rangle$ is a sufficient and necessary condition for the success of~$Q$ when $D$~is odd,
which we will prove later.

To sum up, the existence of such an entity~$c$ (when $D$~is even) or a pair of entities $\langle c,c' \rangle$ (when $D$~is odd)
is a necessary and sufficient condition for the success of~$Q$.
We call~$c$ or $\langle c,c' \rangle$ a \emph{certificate} of the success of~$Q$.

\subsubsection{Formal Definition}
Theorem~\ref{THE:CERT} formalizes the above idea.
A constructive proof is presented.

\begin{theorem}[certificate]\label{THE:CERT}
  A relationship query~$Q$ submitted to an entity-relation graph~$G$ under a diameter constraint~$D$
  is successful if and only if $\exists c \in E_G$ such that
  \begin{enumerate}
    \item $\forall qe \in Q$, $\dist(qe, c) \leq \left\lceil{\frac{D}{2}}\right\rceil$, and
    \item if $D$~is odd and $\exists qe \in Q$ such that $\dist(qe, c) = \left\lceil{\frac{D}{2}}\right\rceil$, then $c$~has a neighbor~$c' \in \nb(c)$ such that
    $\forall qe \in Q$
    that satisfies $\dist(qe, c) = \left\lceil{\frac{D}{2}}\right\rceil$,
    $\dist(qe, c') = \left\lceil{\frac{D}{2}}\right\rceil - 1$.
  \end{enumerate}
  Such an entity~$c$ (when $D$~is even) or a pair of entities $\langle c,c' \rangle$ (when $D$~is odd) is called a \emph{certificate} of the success of~$Q$.
  For convenience, whenever $D$~is even or odd, we consistently refer to~$c$ as a certificate.
  Given a certificate~$c$, $qe \in Q$ that satisfies $\dist(qe, c) = \left\lceil{\frac{D}{2}}\right\rceil$ is called a \emph{critical query entity}.
\end{theorem}
\begin{proof}
We present a constructive proof.
Recall that the distance between two entities in~$G$ is denoted by $\dist$.
We denote their distance in an SA~$x$ by $\dist_x$.

\textbf{Proof of Necessity.}
Let~$c$ be a central vertex of an SA $x = \langle E_x,A_x \rangle$ that is a result of~$Q$.
We will show that $c$~satisfies the two conditions in the theorem.

For the first condition,
recall that $x$~is a tree,
which satisfies $\rad(x) = \left\lceil{\frac{\diam(x)}{2}}\right\rceil \leq \left\lceil{\frac{D}{2}}\right\rceil$.
Accordingly, $\forall qe \in Q \subseteq E_x$,
$\dist_x(qe,c) \leq \left\lceil{\frac{D}{2}}\right\rceil$.
Immediately we have \\
$\dist(qe,c) \leq \left\lceil{\frac{D}{2}}\right\rceil$
because $x$~is a subgraph of~$G$.

For the second condition,
as $x$~is a tree,
let~$p_{qe}$ be the unique (and hence shortest) path between~$c$ and a critical query entity~$qe$ in~$x$.
The length of~$p_{qe}$ is exactly~$\left\lceil{\frac{D}{2}}\right\rceil$;
it is not larger than~$\left\lceil{\frac{D}{2}}\right\rceil$ because
$\dist_x(qe,c) \leq \left\lceil{\frac{D}{2}}\right\rceil$;
it is not smaller than~$\left\lceil{\frac{D}{2}}\right\rceil$ because
$\dist(qe, c) = \left\lceil{\frac{D}{2}}\right\rceil$.
Therefore, $p_{qe}$~is a shortest path in~$G$.
Let~$c'$ be $c$'s neighbor in~$p_{qe}$.
We will show that $c'$~satisfies the second condition.
If $qe$~is the only critical query entity,
the second condition will be trivially satisfied.
Otherwise, for any critical query entity~$qe'$ in~$x$ other than~$qe$,
the unique and shortest path between~$c$ and~$qe'$ in~$x$
(which is also a shortest path in~$G$)
also passes through~$c'$ because otherwise,
we would have $\dist_x(qe,qe') = \left\lceil{\frac{D}{2}}\right\rceil + \left\lceil{\frac{D}{2}}\right\rceil > D$ when $D$~is odd,
violating the diameter constraint.
Therefore, $\dist(qe',c') = \left\lceil{\frac{D}{2}}\right\rceil - 1$.

\textbf{Proof of Sufficiency.}
When $D$~is even,
for each query entity $qe \in Q$, we choose a shortest path between~$qe$ and~$c$ in~$G$.
All of these paths are merged into a connected subgraph~$x$, which clearly satisfies $\diam(x) \leq D$.
In particular, if the shortest paths between two vertices are not unique,
we consistently choose a particular one of them in a deterministic way to avoid cycles.
It can be determined with the help of a fixed order of the arcs in~$A_G$, e.g., in alphabetical order of their IDs.
However, the specific order is not important.
This ensures that $x$~is a tree and hence is minimal.
Therefore, $x$~is an SA and $Q$~is successful.

When $D$~is odd, we construct a minimal connected subgraph~$x$ in a similar way.
In particular, when $qe \in Q$ is a critical query entity,
we firstly choose a shortest path between~$qe$ and~$c'$ in~$G$,
and then merge that path with the arc between~$c'$ and~$c$ to form a shortest path between~$qe$ and~$c$.
This ensures that $\diam(x) \leq D$.
Therefore, $x$~is an SA and $Q$~is successful.
\end{proof}

Certificates may not be unique.
In Fig.~\ref{fig:erg}, under $D=4$,
the relationship query $\{\texttt{Alice}, \texttt{Dan}\}$ is successful.
Both \texttt{Paper01} and \texttt{ISWC2019} are certificates of its success.

\begin{figure}[!t]
\begin{algorithmic}[1]
  \REQUIRE An entity-relation graph~$G$, a diameter constraint~$D$, and a relationship query~$Q$.
  \ENSURE $Q_{max}$ --- a maximum successful sub-query of~$Q$.
  \STATE $qu \leftarrow$ empty queue
  \STATE $visited \leftarrow \emptyset$
  \FORALL{$qe \in Q$}
    \STATE $qu$.Enqueue($\langle qe,qe \rangle$)
    \STATE $visited \leftarrow visited \cup \{qe\}$
  \ENDFOR
  \STATE $Q_{max} \leftarrow \emptyset$
  \WHILE{$qu$ is not empty}
    \STATE $\langle e,sqe \rangle \leftarrow qu$.Dequeue()
    \STATE $Q_e \leftarrow$ OptWithCert($G$, $D$, $Q$, $e$, $Q_{max}$)
    \IF{$|Q_e| > |Q_{max}|$}
      \STATE $Q_{max} \leftarrow Q_e$
    \ENDIF
    \IF{$\dist(e,sqe) < \left\lceil{\frac{D}{2}}\right\rceil$}
      \FORALL{$e' \in \nb(e)$}
        \IF{$e' \notin visited$}
          \STATE $qu$.Enqueue($\langle e',sqe \rangle$)
          \STATE $visited \leftarrow visited \cup \{e'\}$
        \ENDIF
      \ENDFOR
    \ENDIF
  \ENDWHILE
  \RETURN $Q_{max}$
\end{algorithmic}
\caption{CertQR: a certificate-based algorithm for query relaxation.}
\label{alg:basic}
\end{figure}

\subsection{Certificate-based Algorithm}\label{sect:bscalg}
Based on the notion of certificate, we propose an algorithm for query relaxation.

\begin{figure}[!t]
\begin{algorithmic}[1]
  \REQUIRE {An entity-relation graph~$G$, a diameter constraint~$D$, a relationship query~$Q$, an entity~$c \in E_G$,
  and a known successful sub-query~$Q_{known}$ (or~$\emptyset$ if not known).}
  \ENSURE {$Q_c$ --- a maximum sub-query of~$Q$ that $c$~is a certificate of the success of.
  The algorithm outputs~$\emptyset$ when: $c$~is not a certificate of the success of any sub-query, or no such sub-query is larger than~$Q_{known}$.}
  \STATE $Q^1 \leftarrow \{qe \in Q : \dist(qe,c) \leq \left\lceil{\frac{D}{2}}\right\rceil\}$
  \STATE $Q^2 \leftarrow \{qe \in Q : \dist(qe,c) = \left\lceil{\frac{D}{2}}\right\rceil\}$
  \IF{$|Q^1| > |Q_{known}|$}
      \IF{$D$~is even \OR $|Q^2| \leq 1$}
      \STATE $Q_c \leftarrow Q^1$
    \ELSE
      \FORALL{$c' \in \nb(c)$}
        \STATE $R_{c'} \leftarrow \{qe \in Q^2 : \dist(qe, c') = \left\lceil{\frac{D}{2}}\right\rceil - 1\}$
      \ENDFOR
      \STATE $Q_c \leftarrow (\argmax_{R_{c'}}{|R_{c'}|}) \cup (Q^1 \setminus Q^2)$
    \ENDIF
    \IF{$|Q_c| > 1$}
      \RETURN $Q_c$
    \ELSE
      \RETURN $\emptyset$
    \ENDIF
  \ELSE
    \RETURN $\emptyset$
  \ENDIF
\end{algorithmic}
\caption{The OptWithCert algorithm.}
\label{alg:owc}
\end{figure}

\subsubsection{Algorithm Design}
The problem of query relaxation formulated in Definition~\ref{def:qr}
is to find a maximum successful sub-query of~$Q$, denoted by~$Q_{max}$.
Following Theorem~\ref{THE:CERT},
it turns into finding a maximum sub-query whose success has a certificate.
Such a certificate, if it exists,
is at most $\left\lceil{\frac{D}{2}}\right\rceil$~hops away from every query entity in~$Q_{max}$.
Therefore, we can exhaustively search all the entities
that are at most $\left\lceil{\frac{D}{2}}\right\rceil$~hops away from each query entity in~$Q$.
Each of these entities may be a certificate of the success of one or more sub-queries of~$Q$.
A maximum one of these sub-queries will be~$Q_{max}$.

The algorithm, called CertQR, is presented in Fig.~\ref{alg:basic}.
Breadth-first search simultaneously starts from each query entity $qe \in Q$ (lines~3--6).
The frontier is stored in a queue denoted by $qu$,
where each element is an ordered pair of entities $\langle e,sqe \rangle$,
consisting of an entity~$e \in E_G$ to process,
and a query entity $sqe \in Q$ starting from which $e$~is visited for the first time
(i.e., $sqe$~is the closest query entity to~$e$).
Visited entities are stored in a set denoted by $visited$.
Iteratively, for each entity~$e$ to process (line~9),
the OptWithCert algorithm in Fig.~\ref{alg:owc} finds
$Q_e$,
which is a maximum sub-query that $e$~is a certificate of the success of (line~10).
If $Q_e$~is larger than the current~$Q_{max}$, a substitution will be made to update~$Q_{max}$ (lines~11--13).
If $e$~is less than $\left\lceil{\frac{D}{2}}\right\rceil$~hops away from~$sqe$,
search will continue and expand the neighbors of~$e$, i.e.,~$\nb(e)$ (lines~14--21).
Finally, $Q_{max}$~is returned,
which is either a maximum successful sub-query of~$Q$ if it exists, or~$\emptyset$ (line~23).

OptWithCert in Fig.~\ref{alg:owc} finds~$Q_c$ ---
a maximum sub-query that a particular entity~$c$ is a certificate of the success of.
It returns~$\emptyset$ if $c$~is not a certificate of the success of any sub-query.
Let~$Q^1$ and~$Q^2$ be the sets of query entities that
are at most and are exactly $\left\lceil{\frac{D}{2}}\right\rceil$~hops away from~$c$, respectively (lines~1--2).
According to Theorem~\ref{THE:CERT}, $Q^1$~is exactly~$Q_c$ when $D$~is even (lines~4--5).
When $D$~is odd, we need to consider critical query entities, i.e.,~$Q^2$ (lines~6--11).
Specifically, for each neighbor of~$c$ denoted by $c' \in \nb(c)$,
let~$R_{c'}$ be the subset of critical query entities that
are $(\left\lceil{\frac{D}{2}}\right\rceil - 1)$~hops away from~$c'$ (lines~7--9).
These critical query entities (i.e.,~$R_{c'}$) and all the non-critical query entities (i.e.,~$Q^1 \setminus Q^2$) together
form a maximum sub-query that $\langle c,c' \rangle$~is a certificate of the success of.
$Q_c$~is a maximum one of these sub-queries over all $c' \in \nb(c)$ (line~10).

Two additional improvements are made in OptWithCert.
First, even if $D$~is odd, it will be unnecessary to look for~$c'$ and we will have $Q_c=Q^1$ if $|Q^2| \leq 1$ (lines~4--5),
because either there is no critical query entity (i.e.,~$|Q^2|=0$)
or there is only one critical query entity (i.e.,~$|Q^2|=1$)
so that there certainly exists $c' \in \nb(c)$ that satisfies $R_{c'}=Q^2$.
Second, because a maximum successful sub-query found by the algorithm is bounded by~$Q^1$,
the algorithm will be terminated early if $Q^1$~is not larger than~$Q_{known}$,
which is a known successful sub-query (line~3).
In the CertQR algorithm (line~10),
the current maximum successful sub-query is assigned to $Q_{known}$ when invoking OptWithCert.

\subsubsection{Algorithm Analysis}
The correctness of the CertQR algorithm is straightforward following Theorem~\ref{THE:CERT}.

The running time of CertQR mainly consists of:
\begin{itemize}
    \item $O(|E_G|+|A_G|)$ for breadth-first search, and
    \item the time for $O(|E_G|)$ times of invoking OptWithCert.
\end{itemize}
\noindent The running time of OptWithCert is dominated by distance calculation ($\dist$).
Let~$d$ be the running time for one $\dist$ call.
Altogether, the $O(|E_G|)$ times of invoking OptWithCert
use $O(|E_G| \cdot |Q|d)$ time when $D$~is even (lines~1--2),
and need additional $O(|A_G| \cdot |Q|d)$ time when $D$~is odd (lines~7--10).

Overall, the running time of CertQR is bounded by $O((|E_G|+|A_G|) \cdot |Q|d)$, which is polynomial.
It is more scalable than the intuitive exponential-time solution discussed in Section~\ref{sect:int}.

\textbf{Calculation of Distance.}
The running time of CertQR is proportional to~$d$.
When $G$~is large, online calculating $\dist$ is time-consuming (though still in polynomial time),
and materializing offline calculated distances between all pairs of entities is space-consuming.
To achieve a trade-off between time and space, we implement an off-the-shelf \emph{distance oracle}~\cite{oracle}.
This data structure, based on certain precomputed and materialized information,
allows reasonably fast distance calculation,
though not as fast as directly looking up materialized distances.
Therefore, $d$~can be practically regarded as a constant.
The size of a distance oracle is considerably smaller than
the size of materializing distances between all pairs of entities.
The reader is referred to~\cite{oracle} for a detailed implementation,
and to~\cite{oraclesurvey} for a survey of related techniques.

For the completeness of this article, we briefly describe our implementation of distance oracle according to~\cite{oracle}. For each vertex, its distances to a small set of landmark vertices are precomputed and materialized. For any two vertices, they have at least one common landmark vertex that is on a shortest path between them. Therefore, the distance between two vertices can be quickly calculated based on their materialized distances to their common landmark vertices. Minimizing the number of landmark vertices and materialized distances is the focus of~\cite{oracle}, which we will not detail here.

\subsubsection{Running Example}
Under $D=4$,
for the relationship query $\{\texttt{Alice},\texttt{Bob},\\
\texttt{Dan}, \texttt{Gary}\}$
submitted to~$G$ in Fig.~\ref{fig:erg},
every entity in the graph is at most 2~hops away from some query entity.
So all the 11~entities are visited in search and processed by OptWithCert.
As a result,
\texttt{ISWC2019} is a certificate of the success of $\{\texttt{Alice}, \texttt{Bob}, \texttt{Dan}\}$,
and \texttt{Paper02} is a certificate of the success of $\{\texttt{Bob}, \texttt{Dan}, \texttt{Gary}\}$.
The two sub-queries are equally large.
Either of them will be returned as~$Q_{max}$.
\section{\texorpdfstring{C\MakeLowercase{ert}QR+}{CertQR+}: Improved Algorithm}\label{sect:opt}
CertQR exhaustively considers all the entities
that are at most $\left\lceil{\frac{D}{2}}\right\rceil$ hops away from each query entity as candidate certificates.
To improve the performance,
we propose to consider fewer entities but ensure that
unvisited entities cannot be a certificate of the success of a sub-query that is larger than the one to return.
To achieve it, we devise a \emph{best-first search} algorithm called CertQR+.

We will present the CertQR+ algorithm in Section~\ref{sect:optalg},
and introduce two fine-grained heuristics to further improve the performance in Section~\ref{sect:heu}.

\begin{figure}[!t]
\begin{algorithmic}[1]
  \REQUIRE{An entity-relation graph~$G$, a diameter constraint~$D$, and a relationship query~$Q$.}
  \ENSURE{$Q_{max}$ --- a maximum successful sub-query of~$Q$.}
  \STATE $pq \leftarrow$ empty priority queue
  \FORALL{$qe \in Q$}
    \STATE $pr \leftarrow \pr(qe|qe)$
    \STATE $pq$.InsertWithPriority($\langle qe,qe,pr \rangle$)
    \STATE $visited_{qe} \leftarrow \{qe\}$
  \ENDFOR
  \STATE $Q_{max} \leftarrow \emptyset$
  \STATE $checked \leftarrow \emptyset$
  \WHILE{$pq$ is not empty}
    \STATE $\langle e,sqe,pr \rangle \leftarrow pq$.PullHighestPriorityElement()
    \IF{$\left\lfloor{pr}\right\rfloor \leq |Q_{max}|$ \OR $\left\lfloor{pr}\right\rfloor \leq 1$}
      \STATE break the while loop
    \ELSE
      \IF{$e \notin checked$}
        \STATE $Q_e \leftarrow$ OptWithCert($G$, $D$, $Q$, $e$, $Q_{max}$)
        \STATE $checked \leftarrow checked \cup \{e\}$
        \IF{$|Q_e| > |Q_{max}|$}
          \STATE $Q_{max} \leftarrow Q_e$
        \ENDIF
      \ENDIF
      \IF{$\dist(e,sqe)<\left\lceil{\frac{D}{2}}\right\rceil$ \AND $\left\lfloor{pr}\right\rfloor>|Q_{max}|$}
        \FORALL{$e' \in \nb(e)$}
          \IF{$e' \notin visited_{sqe}$ \AND $\dist(e',sqe)=\dist(e,sqe)+1$}
            \STATE $pr' \leftarrow \pr(e'|sqe)$
            \STATE $pq$.InsertWithPriority($\langle e',sqe,pr' \rangle$)
            \STATE $visited_{sqe} \leftarrow visited_{sqe} \cup \{e'\}$
          \ENDIF
        \ENDFOR
      \ENDIF
    \ENDIF
  \ENDWHILE
  \RETURN $Q_{max}$
\end{algorithmic}
\caption{CertQR+: an improved version of CertQR.}
\label{alg:opt}
\end{figure}

\subsection{Improved Algorithm}\label{sect:optalg}
Our improved algorithm performs best-first search.

\subsubsection{Algorithm Design}
The new algorithm, called CertQR+, is presented in Fig.~\ref{alg:opt}.
Similar to CertQR, search simultaneously starts from each query entity.
However, different from CertQR which visits each entity in~$G$ at most once,
CertQR+ runs $|Q|$~\emph{independent searches}.
Each search is only focused on sub-queries that contain the start query entity of that search.
An entity in the graph may be visited up to $|Q|$~times in $|Q|$~searches.
This seems slower than CertQR,
but we will show that these independent searches can be terminated early.

Specifically, each search starts from a distinct query entity $qe \in Q$ (lines~2--6).
The frontier, which is shared by all the $|Q|$~searches, is stored in a priority queue denoted by $pq$,
where each element is an ordered entity-entity-priority triple $\langle e,sqe,pr \rangle$,
consisting of an entity $e \in E_G$ to process,
a query entity $sqe \in Q$ starting from which $e$~is visited
(i.e., $sqe$~identifies a search),
and a priority for~$e$ denoted by~$pr$.
Entities visited in a search starting from~$qe$ are stored in a set denoted by $visited_{qe}$.

Different from CertQR which uses a first-in-first-out queue and performs breadth-first search,
CertQR+ uses a priority queue~$pq$ and performs best-first search.
In each iteration, it pulls out a triple $\langle e,sqe,pr \rangle$
that has the highest priority~$pr$ (line~10).
We define~$pr$ to be:
an \emph{estimate} of the number of query entities in a maximum possible sub-query
which \emph{$e$~or its descendant in the search starting from~$sqe$} can be a certificate of the success of.
If the highest priority in~$pq$ is not larger than~$|Q_{max}|$ or~1
where $Q_{max}$ denotes the current maximum successful sub-query,
the algorithm can be terminated and $Q_{max}$~will be returned (lines~11-12).
We will later elaborate the computation of priority.

Although $e$~may be visited in different independent searches,
it will be processed by OptWithCert at most once (lines~14--16).
Entities that have been processed by OptWithCert are stored in a set denoted by $checked$,
which is shared by all the $|Q|$~searches.
If OptWithCert finds a larger successful sub-query, $Q_{max}$~will be updated (lines~17--19).

The search starting from~$sqe$ will expand the neighbors of~$e$ if:
$e$~is less than $\left\lceil{\frac{D}{2}}\right\rceil$~hops away from~$sqe$,
and $pr$~is larger than $Q_{max}$ (lines~21--29).
A neighbor $e' \in \nb(e)$ will be expanded only if $\dist(e',sqe)=\dist(e,sqe)+1$ (line~23),
i.e., $e'$~is reached via a shortest path from~$sqe$.
This additional requirement is not necessary but may reduce the search space, which we will discuss later.

Finally, $Q_{max}$~is returned,
which is either a maximum successful sub-query of~$Q$ if it exists, or~$\emptyset$ (line~32).

\textbf{Computation of Priority.}
The priority for an entity~$e$, i.e.,~$pr$,
is computed by a function $\pr(e|sqe)$ which depends on~$sqe$:
\begin{equation}\label{eq:prip}
  \pr(e|sqe) = |\est(e|sqe)| \,,
\end{equation}
\noindent where $\est(e|sqe)$ is an estimated set of query entities
in a maximum possible sub-query
which $e$~or its descendant in the search starting from~$sqe$ can be a certificate of the success of.
Our estimation uses distances between entities:
\begin{equation}\label{eq:est}
\begin{split}
  \est(e|sqe) = \{sqe\} \cup \{ & qe \in (Q \setminus \{sqe\}) : \\ 
  & \dist(e,sqe) + \dist(e,qe) \leq D \} \,.
\end{split}
\end{equation}
\noindent We will prove that it guarantees the optimality of~$Q_{max}$ when the algorithm is terminated.

In Section~\ref{sect:heu}, we will consider more effective implementation of $\pr(e|sqe)$ based on fine-grained heuristics.

\subsubsection{Algorithm Analysis}
Theorem~\ref{THE:OPT} proves the correctness of CertQR+.
\begin{theorem}\label{THE:OPT}
CertQR+ returns a maximum successful sub-query if it exists.
\end{theorem}

The key idea of our proof is to show
the existence of a path between some query entity in an optimum solution and a certificate of its success such that:
for every entity in the path, its priority is not smaller than the number of query entities in that optimum solution.
Consequently, as search starts from each query entity,
it is impossible that the algorithm produces a sub-optimum solution without visiting that certificate via the path.

\begin{proof}
We prove by contradiction.

Assume CertQR+ returns a sub-optimum solution~$Q_{max}$
which is smaller than an optimum solution denoted by~$Q_{opt}$.
Let $c \in E_G$ be a certificate of the success of~$Q_{opt}$.
Before CertQR+ is terminated,
$c$~has never been processed by OptWithCert because otherwise, $Q_{opt}$ instead of~$Q_{max}$ would be returned.
However, we can prove the existence of a particular path between~$c$ and some query entity $sqe \in Q_{opt}$, called a \emph{key path}, such that:
for every entity~$e$ in this path,
$\pr(e|sqe) \geq |Q_{opt}| > |Q_{max}|$ holds.
With this key path, the algorithm is impossible to return~$Q_{max}$ without visiting~$c$ via this path and finding~$Q_{opt}$ by OptWithCert, leading to a contradiction.

To show the existence of such a key path,
consider an SA~$x$ constructed according to the proof of sufficiency of Theorem~\ref{THE:CERT},
which consists of shortest paths between~$c$ and each query entity in~$Q_{opt}$.
Now we prove that:
(i)~at least one of these shortest paths is not longer than $\left\lfloor{\frac{D}{2}}\right\rfloor$, and (ii)~this shortest path is a key path.

(i)~When $D$~is even,
all of those shortest paths are not longer than $\left\lceil{\frac{D}{2}}\right\rceil = \left\lfloor{\frac{D}{2}}\right\rfloor$.
When $D$~is odd,
assume on the contrary that
all of them are longer than $\left\lfloor{\frac{D}{2}}\right\rfloor$.
According to Theorem~\ref{THE:CERT},
all the query entities are critical query entities,
and hence all of those shortest paths pass through~$c'$ which is a neighbor of~$c$.
That contradicts the minimality of~$x$ in Definition~\ref{def:sa} because
the vertex~$c$ and the arc between~$c$ and~$c'$ can be removed from~$x$ to obtain a proper subgraph of~$x$ that is a well-defined SA.

(ii)~Let~$p$ be a path not longer than $\left\lfloor{\frac{D}{2}}\right\rfloor$ proved in~(i),
which connects~$sqe \in Q_{opt}$ and~$c$.
For every entity~$e$ in~$p$ and every query entity $sqe' \in (Q_{opt} \setminus \{sqe\})$,
$\dist(e,sqe) + \dist(e,sqe')$ is not larger than
the sum of the length of~$p$ and the length of a shortest path between~$c$ and~$sqe'$,
which in turn is not larger than $\left\lfloor{\frac{D}{2}}\right\rfloor + \left\lceil{\frac{D}{2}}\right\rceil = D$.
Therefore, $sqe' \in \est(e|sqe)$
and we have $\pr(e|sqe) = |\est(e|sqe)| \geq |Q_{opt}|$, so $p$~is a key path.
\end{proof}

In order to find an optimum solution,
it would be sufficient to visit a key path which is a shortest path.
It explains why CertQR+ only searches along shortest paths (line~23).

The running time of CertQR+ mainly consists of:
\begin{itemize}
    \item $O(|Q| \cdot (|E_G|+|A_G|))$ for $|Q|$~searches,
    \item $O((|E_G|+|A_G|) \cdot |Q|d)$ for $O(|E_G|)$~times of invoking OptWithCert,
    being the same as that in CertQR,
    \item $O(|Q|^2d \cdot |E_G|)$ for $O(|Q| \cdot |E_G|)$~times of priority computation, and
    \item $O(|Q| \cdot |E_G| \log (|Q| \cdot |E_G|))$ for $O(|Q| \cdot |E_G|)$~pairs of insert-pull operations using a Fibonacci heap for priority queue.
\end{itemize}

The overall worst-case running time of CertQR+ can be longer than that of CertQR.
However, CertQR+ could be faster in practice
mainly because searches in CertQR+ can be terminated early.

\subsubsection{Running Example}
Under $D=4$, for the relationship query $Q=\{\texttt{Alice},\\ \texttt{Bob},\texttt{Dan},\texttt{Gary}\}$
submitted to~$G$ in Fig.~\ref{fig:erg},
CertQR+ initially inserts four triples into~$pq$:
\begin{equation*}
\begin{split}
    t_1 & = \langle \texttt{Alice},\texttt{Alice},3 \rangle \,,\\
    t_2 & = \langle \texttt{Bob},\texttt{Bob},4 \rangle \,,\\
    t_3 & = \langle \texttt{Dan},\texttt{Dan},4 \rangle \,,\\
    t_4 & = \langle \texttt{Gary},\texttt{Gary},3 \rangle \,.
\end{split}
\end{equation*}

In~$pq$, $t_2$ (or~$t_3$) has the highest priority and is pulled out first.
However, OptWithCert finds that \texttt{Bob} is not a certificate of the success of any sub-query.
\texttt{Bob}'s neighbors are expanded, and a new triple is inserted into~$pq$:
\begin{equation*}
    t_5 = \langle \texttt{Paper02},\texttt{Bob},4 \rangle \,.
\end{equation*}

$t_5$ (or~$t_3$) is pulled out of~$pq$.
OptWithCert finds that \texttt{Paper02} is a certificate of the success of $\{\texttt{Bob},\texttt{Dan},\texttt{Gary}\}$,
which is assigned to~$Q_{max}$.
\texttt{Paper02}'s neighbors are then expanded, and three new triples are inserted into~$pq$:
\begin{equation*}
\begin{split}
    t_6 & = \langle \texttt{ISWC2019},\texttt{Bob},3 \rangle \,,\\
    t_7 & = \langle \texttt{Erin},\texttt{Bob},1 \rangle \,,\\
    t_8 & = \langle \texttt{Frank},\texttt{Bob},2 \rangle \,.
\end{split}
\end{equation*}

$t_3$~is pulled out of~$pq$.
OptWithCert finds that \texttt{Dan} is not a certificate of the success of any sub-query.
\texttt{Dan}'s neighbors are expanded, and a new triple is inserted into~$pq$:
\begin{equation*}
    t_9 = \langle \texttt{ISWC2019},\texttt{Dan},4 \rangle \,.\\
\end{equation*}
\noindent Note that \texttt{ISWC2019} has been visited twice in two different searches:
one starting from \texttt{Bob} (i.e.~$t_6$) and the other from \texttt{Dan} (i.e.,~$t_9$).

$t_9$~is pulled out of~$pq$.
OptWithCert finds that \texttt{ISWC2019} is not a certificate of the success of any sub-query
that is larger than the current $Q_{max} = \{\texttt{Bob},\texttt{Dan},\texttt{Gary}\}$.
\texttt{ISWC2019}'s neighbors are expanded, and three new triples are inserted into~$pq$:
\begin{equation*}
\begin{split}
    t_{10} & = \langle \texttt{Paper01},\texttt{Dan},2 \rangle \,,\\
    t_{11} & = \langle \texttt{Paper02},\texttt{Dan},3 \rangle \,,\\
    t_{12} & = \langle \texttt{Paper03},\texttt{Dan},1 \rangle \,.
\end{split}
\end{equation*}

Now the highest priority in~$pq$ is~3,
which is not larger than $|Q_{max}|=3$.
The algorithm is terminated and returns $Q_{max} = \{\texttt{Bob},\texttt{Dan},\texttt{Gary}\}$.

For this example, CertQR+ is faster than CertQR.
First, in CertQR, 11~entities are processed by OptWithCert,
whereas only 4~entities are processed by OptWithCert in CertQR+.
Second, \texttt{Carol} is visited in CertQR but not in CertQR+.
CertQR+ effectively prunes the search space of CertQR.

\subsection{Fine-Grained Heuristics}\label{sect:heu}
We introduce two fine-grained heuristics that may guide CertQR+ to find an optimum solution earlier.
In Eq.~(\ref{eq:prip}), $\pr(e|sqe)$ is an integer.
We will heuristically define its fractional part.
That will not affect the correctness of CertQR+,
which only relies on the integer part of priority (lines~11 and~21).
On the other hand, ties in the priority queue (according to the integer part) will be broken not arbitrarily but heuristically (according to the fractional part).
This fine-grained ordering may allow CertQR+ to be terminated earlier.

\subsubsection{Degree-enhanced Priority}
This heuristic considers the \emph{degree} of~$e$, namely the number of arcs in~$G$ that are incident from/to~$e$, denoted by $\degree(e)$:
\begin{equation}
  \degree(e) = |\{a \in A_G : \tail(a)=e \text{ or } \head(a)=e\}| \,.
\end{equation}
\noindent The idea is:
when multiple entities have the same value of $|\est|$,
we will give priority to the one having the smallest degree,
because expanding the neighbors of an entity having a large degree and computing their priorities is time-consuming (lines~22--28).
Therefore, we define a variant of $\pr$ as follows:
\begin{equation}\label{eq:dg}
  \pr_{\text{dg}}(e|sqe) = |\est(e|sqe)| + \frac{1}{2+\degree(e)} \,,
\end{equation}
\noindent where the fractional part depends on~$\degree$.

\subsubsection{Distance-enhanced Priority}
This heuristic considers the distance between~$e$ and possible certificates of the success of $\est(e|sqe)$.
According to Theorem~\ref{THE:CERT},
the distance between such a certificate (if it exists) and each query entity $qe \in \est(e|sqe)$
is not longer than $\left\lceil{\frac{D}{2}}\right\rceil$.
Therefore, the number of hops from~$e$ to a possible certificate is bounded below by:
\begin{equation}
  \hops(e|sqe) = \max_{qe \in \est(e|sqe)}{\max{\{0, \dist(e,qe) - \left\lceil{\frac{D}{2}}\right\rceil\}}} \,.
\end{equation}
\noindent The idea is:
when multiple entities have the same value of $|\est|$,
we will give priority to the closest one to a possible certificate,
which may lead to a better solution earlier.
Therefore, we define a variant of $\pr$ as follows:
\begin{equation}\label{eq:ds}
  \pr_{\text{ds}}(e|sqe) = |\est(e|sqe)| + \frac{1}{2^{1+\hops(e|sqe)}} \,,
\end{equation}
\noindent where the fractional part depends on~$\hops$.

\subsubsection{Combined Priority}
The above two heuristics can be combined as follows:
\begin{equation}\label{eq:dgs}
  \pr_{\text{dgs}}(e|sqe) = |\est(e|sqe)| + \frac{1}{(2+\degree(e))^{1+\hops(e|sqe)}} \,.
\end{equation}
\section{Experiments}\label{sect:eval}
Our code and resources for experiments are available\footnote{\url{https://github.com/nju-websoft/CertQR}}.

At the beginning of the article, we mentioned two methods for query relaxation:
relaxing the compactness constraint and relaxing the query entities.
In the first experiment, we compared the quality of their output,
to show the better practicability of the latter method.
In the second experiment, we tested the running time of our proposed algorithms compared with baselines,
to demonstrate their efficiency.
All the experiments were performed on an Intel Xeon E7-4820 (2.00 GHz) with 128GB memory for Java.

\subsection{Datasets}
Three popular RDF datasets were used in our experiments: DBpedia, LinkedMDB, and Mondial.
They covered the datasets used in previous efforts to evaluate SA search~\cite{tkde17} as well as keyword search on graphs~\cite{emp}.

\textbf{DBpedia} is a large encyclopedic RDF dataset derived from Wikipedia,
describing people, places, creative works, organizations, species, etc.
Using its 2015-10 version\footnote{\url{https://wiki.dbpedia.org/Downloads2015-10}},
we obtained an entity-relation graph by importing relations between entities from two dump files:
\emph{Mappingbased Objects} and \emph{Person data}.

\textbf{LinkedMDB} is a large RDF dataset
describing movies and related concepts, e.g., actors, directors.
We obtained an entity-relation graph by importing relations between entities
from its latest dump file\footnote{\url{http://www.cs.toronto.edu/~oktie/linkedmdb/linkedmdb-latest-dump.zip}}.

\textbf{Mondial} is a small geographical database.
We obtained an entity-relation graph
by importing relations between entities from its RDF graph
version\footnote{\url{http://www.dbis.informatik.uni-goettingen.de/Mondial/Mondial-RDF/mondial.rdf}}.

\begin{table}[!t]
  \caption{Size of Entity-Relation Graphs}
  \label{tab:ds}
  \centering
  \begin{tabular}{lrr}
    \hline
    Dataset & Number of vertices & Number of arcs \\
    \hline
    DBpedia & 5,356,286 & 17,494,749 \\
    LinkedMDB & 1,326,784 & 2,132,796 \\
    Mondial & 8,478 & 34,868 \\
    \hline
\end{tabular}
\end{table}

Note that \texttt{rdf:type} and literals were not included in the entity-relation graphs because our focus was on the relations between instance-level entities. Table~\ref{tab:ds} presents the size of each graph.

\subsection{Queries}
For each entity-relation graph, we constructed two types of queries having complementary characteristics.

A \textbf{simulated query} consisted of related entities, which were likely to be pairwise close in a graph.
We adopted the process described in~\cite{tkde17} to construct such queries
that well simulated real users' information needs.

\begin{figure*}[!t]
\centering
\subfloat[DBpedia (simulated queries)]{\includegraphics[width=0.64\columnwidth]{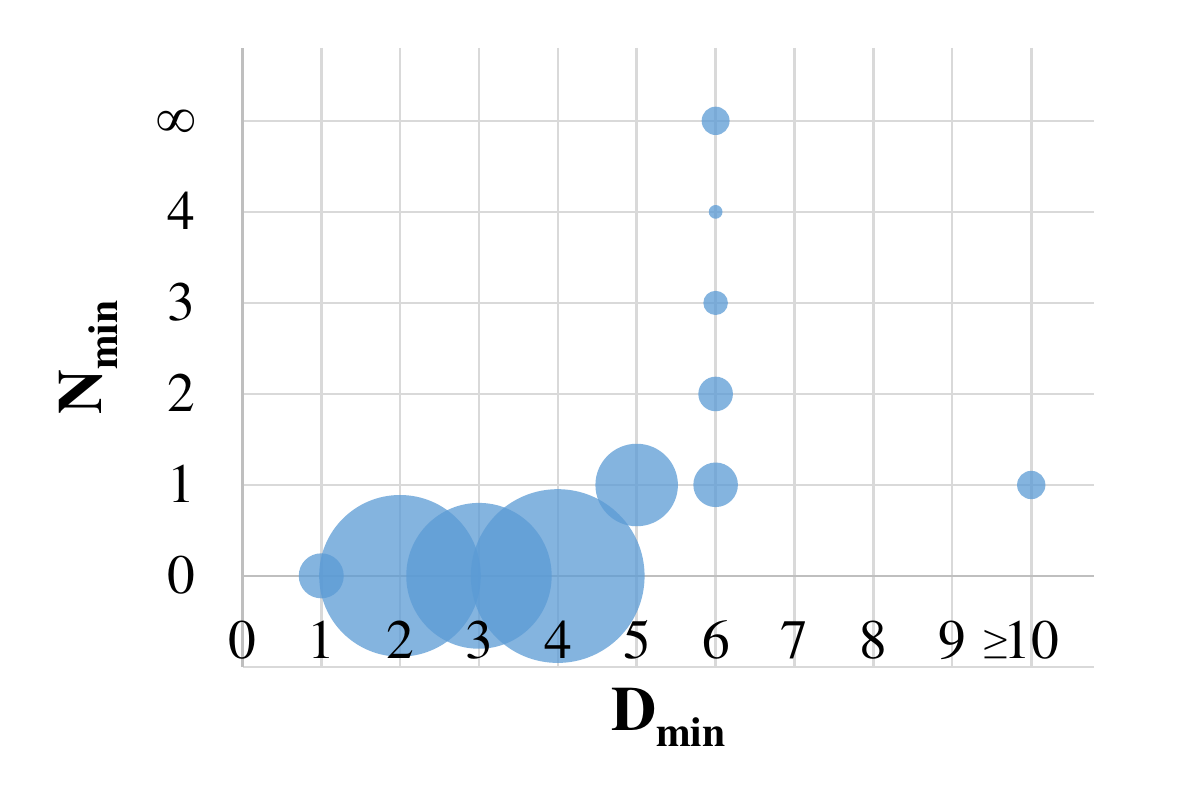}\label{fig:dbparn}}
\hfill
\subfloat[LinkedMDB (simulated queries)]{\includegraphics[width=0.64\columnwidth]{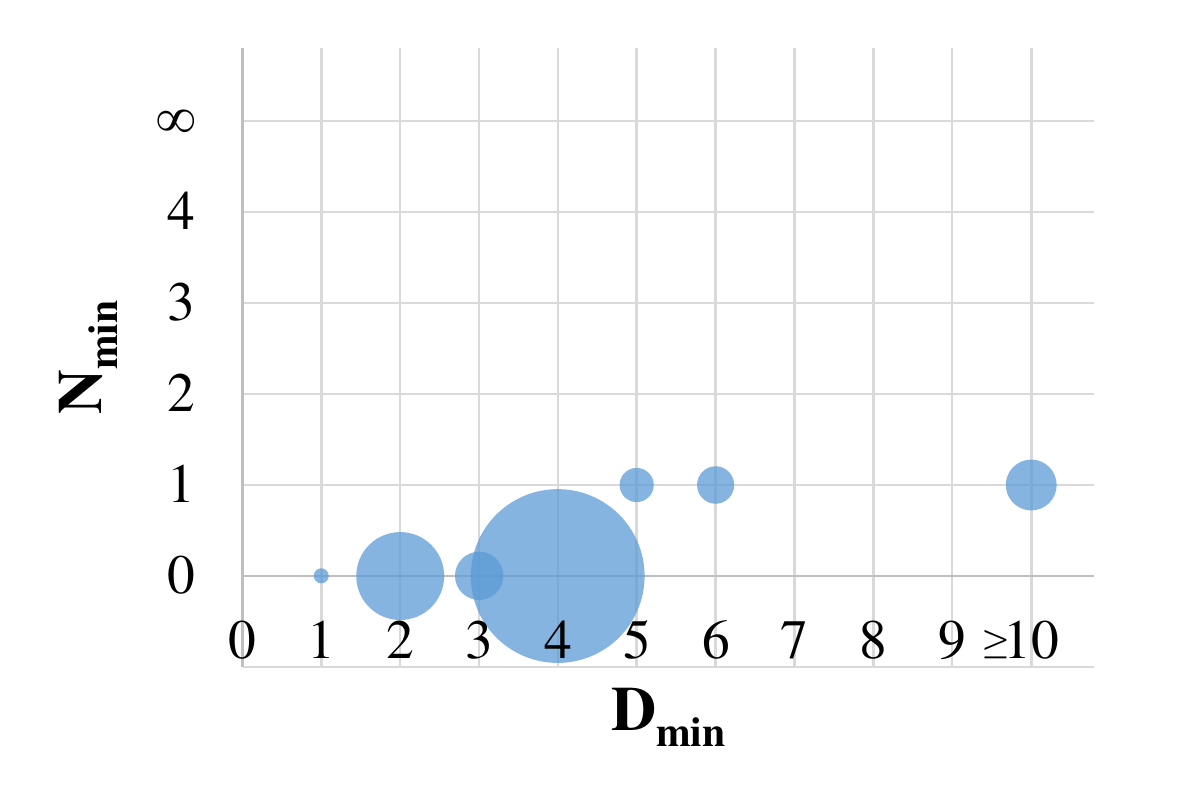}\label{fig:mdbarn}}
\hfill
\subfloat[Mondial (simulated queries)]{\includegraphics[width=0.64\columnwidth]{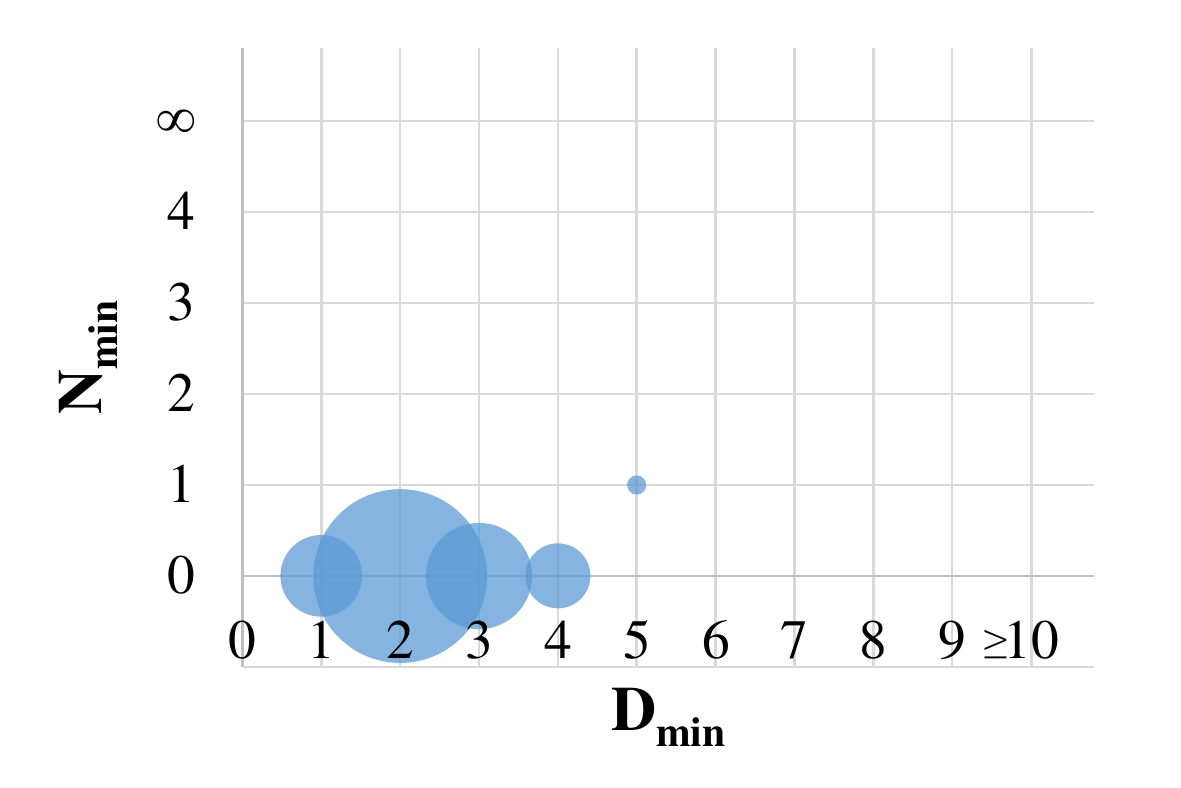}\label{fig:mondarn}}
%
\caption{Distribution of result quality after query relaxation.
The area of a bubble represents the number of queries
having a particular value of $\langle D_{min}, N_{min} \rangle$
represented by the center of the bubble.}
\label{fig:rn}
\end{figure*}

Specifically, for DBpedia, we identified 250~entities
that were mentioned in the training set of the QALD-5 evaluation campaign\footnote{\url{https://qald.sebastianwalter.org/index.php?q=5}},
called \emph{seed entities}.
For each seed entity (e.g., Michael Jordan),
we submitted its name as a keyword query to the Google search engine,
which might trigger Google's Knowledge Graph to return a set of entities
that ``people also search for'' (e.g., Kobe Bryant).
We identified their corresponding entities in DBpedia, if they existed,
called \emph{related entities}.
For 92~seed entities, at least five related entities could be found.
For each of those seed entities and each number of query entities~$n$ in the range of 2--6,
we generated a simulated query consisting of the seed entity
and $n-1$~entities randomly selected from its related entities.

For LinkedMDB and Mondial, the processes were similar.
For LinkedMDB, seed entities were 39~top rated movies in IMDb.
Related entities were movies that
``people who liked (a seed entity) also liked'' recommended by IMDb.
For Mondial, seed entities were 57~entities mentioned in the keyword queries used in a previous evaluation effort\footnote{\url{https://doi.org/10.18130/V3/KEVCF8}}.
Related entities were obtained in the same way we did for DBpedia.

There were $(92+39+57) \cdot 5 = 940$ simulated queries.

A \textbf{random query} consisted of random entities, which were not likely to be pairwise close in a graph.
Complementary to simulated queries, random queries were more likely to fail and need to be relaxed.
For each entity-relation graph and each number of query entities~$n$ in the range of 2--6,
we constructed 100~random queries, each consisting of $n$~randomly selected entities.
There were $100 \cdot 5 \cdot 3 = 1500$ random queries.

\subsection{Experiment on Quality}\label{sect:expq}
In Section~\ref{sect:intro}, we mentioned two methods for query relaxation.
One straightforward approach was to relax the compactness constraint.
The other, as the focus of this article, was to relax the query entities.
In this experiment, we compared the quality of their output to analyze their practicability.
Therefore, we only used simulated queries for this experiment.

\subsubsection{Experiment Design}
We measured quality from two perspectives: compactness and completeness.

\textbf{Compactness.} For each query, we found the minimum diameter constraint under which the query was successful, denoted by~$D_{min}$.
It characterized the \emph{compactness of search results after relaxing the diameter constraint}.
It also indicated the practicability of this approach,
which would be impracticable if $D_{min}$~were large
so that existing algorithms for SA search could not scale well.

\textbf{Completeness.} On the other hand, for each query under $D=4$ which was a typical setting used in the literature~\cite{iswc16,tkde17},
we computed the smallest number of query entities to remove in order to obtain a successful sub-query if it existed, denoted by~$N_{min}$;
otherwise, we defined $N_{min}=\infty$.
It characterized the \emph{completeness of search results after relaxing the query entities},
in terms of the number of missing query entities.

\subsubsection{Experiment Results and Analysis}
Figure~\ref{fig:rn} presents the distribution of~$D_{min}$ and~$N_{min}$.
The area of a bubble represents the number of queries having a particular value of $\langle D_{min}, N_{min} \rangle$
represented by the center of the bubble.

Among all the simulated queries, 87\%~on DBpedia (Fig.~\ref{fig:dbparn}), 89\%~on LinkedMDB (Fig.~\ref{fig:mdbarn}), and 99\%~on Mondial (Fig.~\ref{fig:mondarn}) were successful under $D=4$.
Query relaxation was needed not for them but for the rest of queries that failed,
whose proportions were still considerable on large entity-relation graphs like DBpedia and LinkedMDB.

For 82\%~of the failing queries on DBpedia and 100\%~of the failing queries on LinkedMDB and Mondial,
removing only one query entity by our approach could result in a successful sub-query ($N_{min}=1$).
The loss of result completeness, i.e., the number of missing query entities, was managed at the lowest level.

By contrast, if we chose the alternative approach of relaxing the compactness constraint,
by lifting the diameter constraint (i.e.,~$D$) from~4 to~5,
only 55\%~of the failing queries on DBpedia, 23\%~on LinkedMDB, and 100\%~on Mondial
could turn into successful queries.
In fact, even $D_{min} \geq 10$ was frequently observed when some queries entities were distantly connected or even disconnected in an entity-relation graph.
Therefore, the downside of this alternative approach included:
(a)~the loss of result compactness,
(b)~the incapability to handle disconnected query entities,
and more importantly, (c)~the performance issue considering that
existing techniques for SA search could not scale to large graphs unless $D$~was very small~\cite{iswc16,emp}.
For example, the search algorithm in~\cite{iswc16} could only scale up to~$D=4$ on DBpedia.

To conclude, \emph{relaxing the compactness constraint could not fundamentally solve the problem of query failure.
Relaxing the query entities was a more practicable approach,
which preserved result compactness and minimized result incompleteness.}

\begin{figure}[!t]
\centering
\includegraphics[width=0.95\columnwidth]{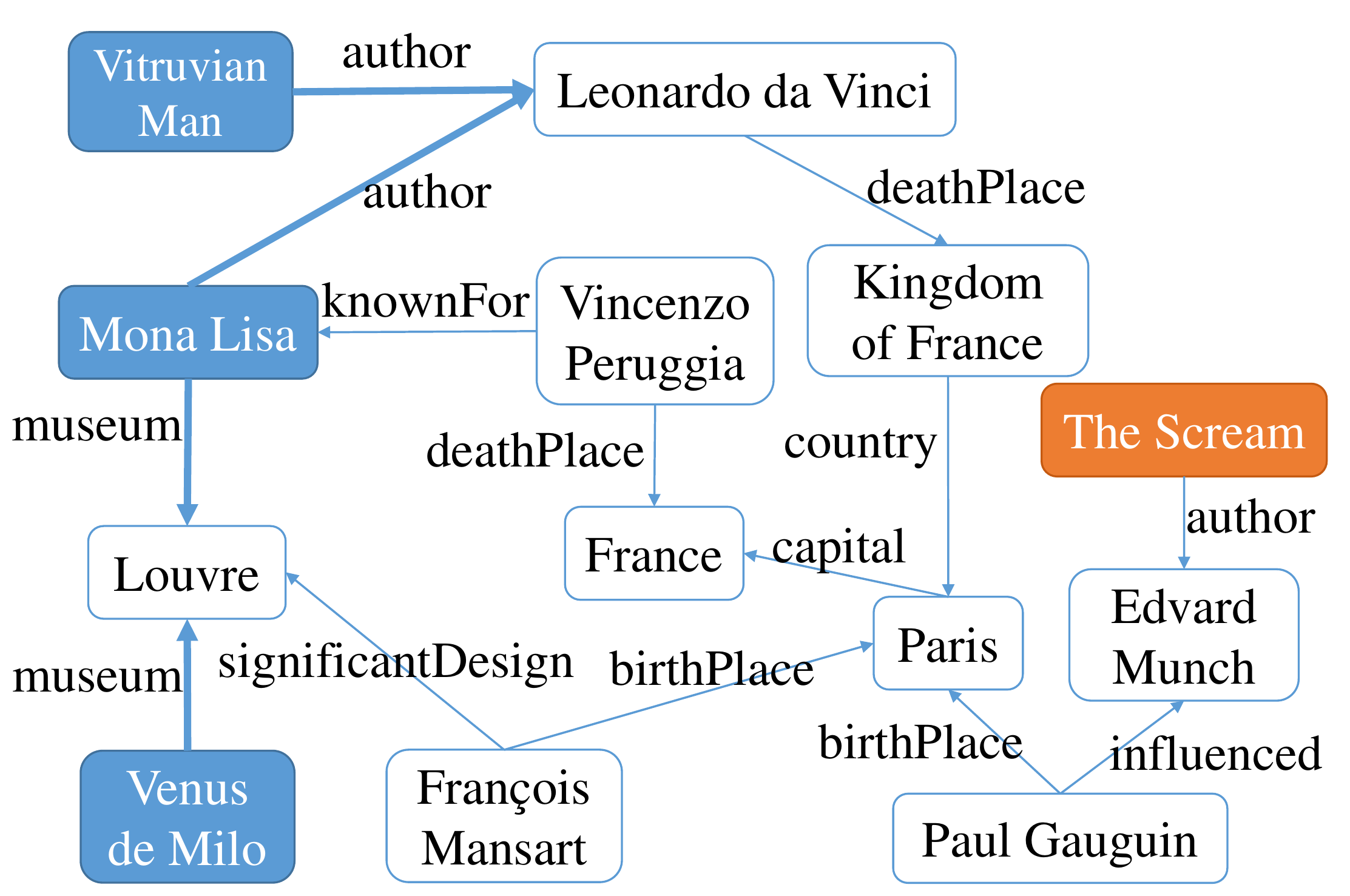}
\caption{A snippet of DBpedia.
Under $D=4$, the query $\{\texttt{Vitruvian Man}, ~\texttt{Mona Lisa}, ~\texttt{Venus de Milo}, ~\texttt{The Scream}\}$
is relaxed to $\{\texttt{Vitruvian Man}, ~\texttt{Mona Lisa}, ~\texttt{Venus de Milo}\}$,
which yields an SA comprising the four bold arcs.}\label{fig:eg-art}
\end{figure}

\begin{table*}[!t]
\caption{Proportion of Timeout Runs on BSL}
\label{tab:tobsl}
\centering
\begin{tabular}{llrrrrr}
    \hline
    Dataset & Query & $n=2$ & $n=3$ & $n=4$ & $n=5$ & $n=6$ \\
    \hline
    DBpedia & Simulated & 2.17\% & 3.80\% & 4.89\% & 10.87\% & 19.29\% \\
    & Random & 0 & 0 & 2.50\% & 12.50\% & 12.00\% \\
    LinkedMDB & Simulated & 0 & 0 & 0.64\% & 3.85\% & 10.90\% \\
    & Random & 0 & 0 & 0 & 0 & 0 \\
    Mondial & Simulated & 0 & 6.14\% & 13.60\% & 36.84\% & 40.79\% \\
    & Random & 0 & 0 & 2.50\% & 14.75\% & 20.50\% \\
    \hline
\end{tabular}
\end{table*}

\begin{table*}[!t]
\caption{Proportion of Timeout Runs on BSL+}
\label{tab:tobslp}
\centering
\begin{tabular}{llrrrrr}
    \hline
    Dataset & Query & $n=2$ & $n=3$ & $n=4$ & $n=5$ & $n=6$ \\
    \hline
    DBpedia & Simulated & 0.54\% & 0.54\% & 3.80\% & 9.78\% & 16.30\% \\
    & Random & 0 & 0 & 2.25\% & 7.25\% & 12.00\% \\
    LinkedMDB & Simulated & 0 & 0 & 7.69\% & 12.18\% & 15.38\% \\
    & Random & 0 & 0 & 0 & 0 & 0 \\
    Mondial & Simulated & 0 & 5.26\% & 9.65\% & 28.51\% & 32.02\% \\
    & Random & 0 & 0 & 1.25\% & 13.50\% & 10.50\% \\
    \hline
\end{tabular}
\end{table*}

\textbf{Case Study.}
One simulated query for DBpedia in our experiments was about four artworks:
\texttt{Vitruvian Man}, \texttt{Mona Lisa}, \texttt{Venus de Milo}, and \texttt{The Scream}.
As illustrated in Fig.~\ref{fig:eg-art},
the query failed under $D=4$ because \texttt{The Scream} was far away from the other three query entities in the graph.
Its maximum successful sub-query was obtained by removing \texttt{The Scream}, resulting in the following SA:
\texttt{Leonardo da Vinci} is the author of both \texttt{Vitruvian Man} and \texttt{Mona Lisa};
the latter is located in the \texttt{Louvre} which also hosts \texttt{Venus de Milo}.
Alternatively, the original query would be successful when $D$~was increased to~6,
to which the algorithm for SA search in~\cite{iswc16} could not scale.
Besides, a resulting large-sized SA would be diffuse,
containing almost all the entities and relations in Fig.~\ref{fig:eg-art}.

\subsection{Experiment on Efficiency}
We tested the running time of our proposed algorithms.
Both simulated and random queries were used for this experiment.

\subsubsection{Participating Algorithms}
We compared 7~algorithms.
We did not find any methods in the literature
that could be directly used to solve our problem,
so we adapted and improved a state-of-the-art algorithm for SA search as a baseline.

\textbf{BSL} was a baseline approach that
checked all possible sub-queries in non-increasing order of the number of query entities they contained,
and returned the first successful one.
To check a sub-query, BSL performed a state-of-the-art path-merging and path-pruning based algorithm for SA search~\cite{iswc16},
which would be terminated early after finding the first SA (thereby indicating success).
It was the first intuitive solution discussed in Section~\ref{sect:int}.

\textbf{BSL+} was a variant of BSL we developed in order to improve the performance.
Considering that BSL repeatedly ran the search algorithm~\cite{iswc16} on the same set of query entities,
we disabled path pruning but cached all the unpruned paths found in its first run,
so that they could be reused in the subsequent runs.

\textbf{CertQR} denoted our certificate-based algorithm presented in Section~\ref{sect:basic}.

\textbf{CertQR+} denoted our improved algorithm presented in Section~\ref{sect:opt},
not using any heuristics introduced in Section~\ref{sect:heu}.

\textbf{dg}, \textbf{ds}, and \textbf{dgs} denoted our improved algorithm presented in Section~\ref{sect:opt}
using the heuristics defined by Eq.~(\ref{eq:dg}), Eq.~(\ref{eq:ds}), and Eq.~(\ref{eq:dgs}), respectively.

For all the algorithms, entity-relation graphs were stored in memory using the JGraphT library\footnote{https://jgrapht.org/}. We used Boolean arrays to record visited entities and those having been processed by OptWithCert, so that set membership could be checked in constant time. The priority queue in CertQR+ was implemented with a heap. BSL and our proposed algorithms used distance oracles, which affordably occupied 9.2GB, 7.2GB, and 0.4GB memory for DBpedia, LinkedMDB, and Mondial, respectively.

\subsubsection{Experiment Design}
For each query, we ran each algorithm under each diameter constraint (i.e.,~$D$) in the range of 3--6.
To obtain more reliable results, we ran each algorithm three times and took the median running time.

We set a timeout of 1,000~seconds.
Any single run of an algorithm would be terminated when reaching timeout.
In that case, the running time was defined to be the timeout value.
Therefore, the longest running time reported in the following was bounded by 1,000~seconds.

\begin{figure*}[!t]
\centering
\subfloat[DBpedia (simulated queries)]{\includegraphics[width=0.6\columnwidth]{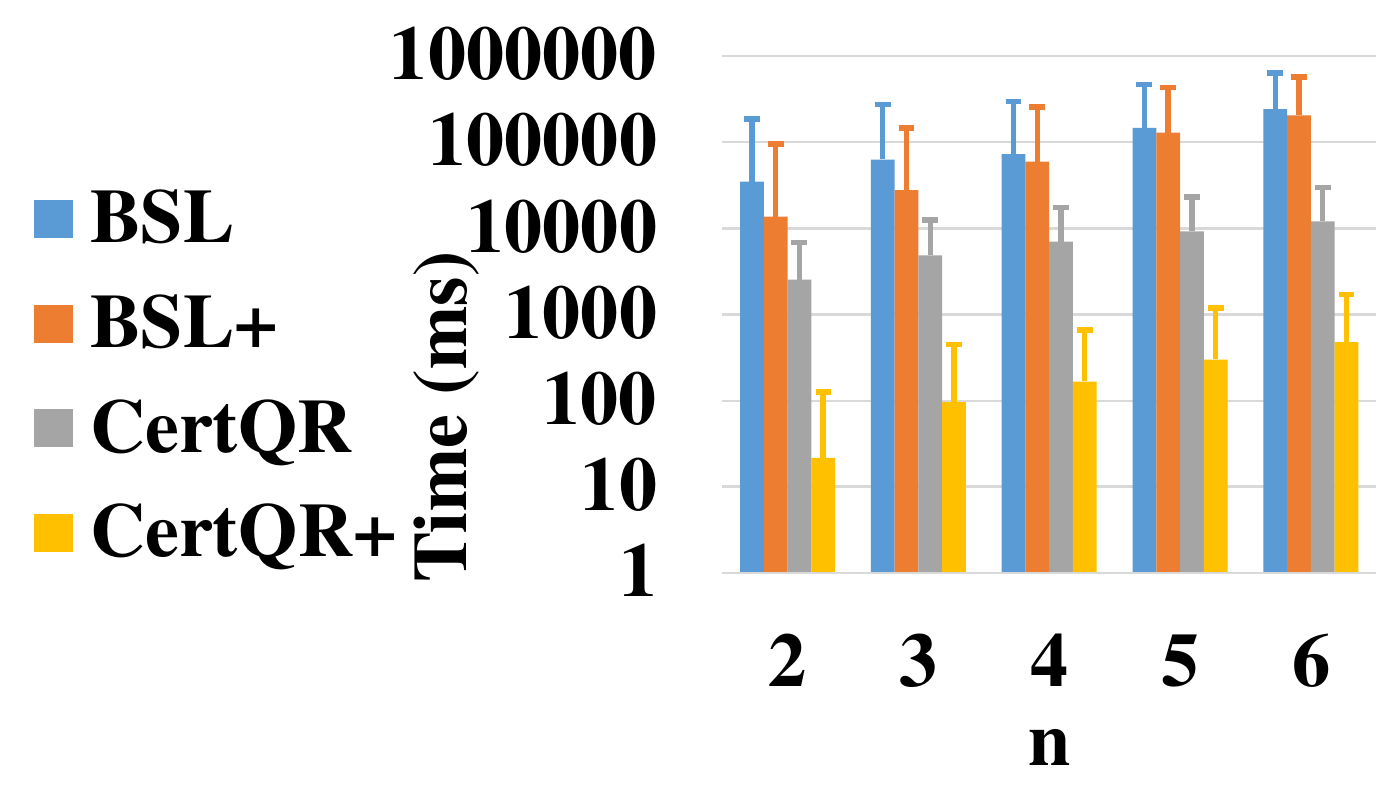}\label{fig:dbpan}}
\hfill
\subfloat[LinkedMDB (simulated queries)]{\includegraphics[width=0.6\columnwidth]{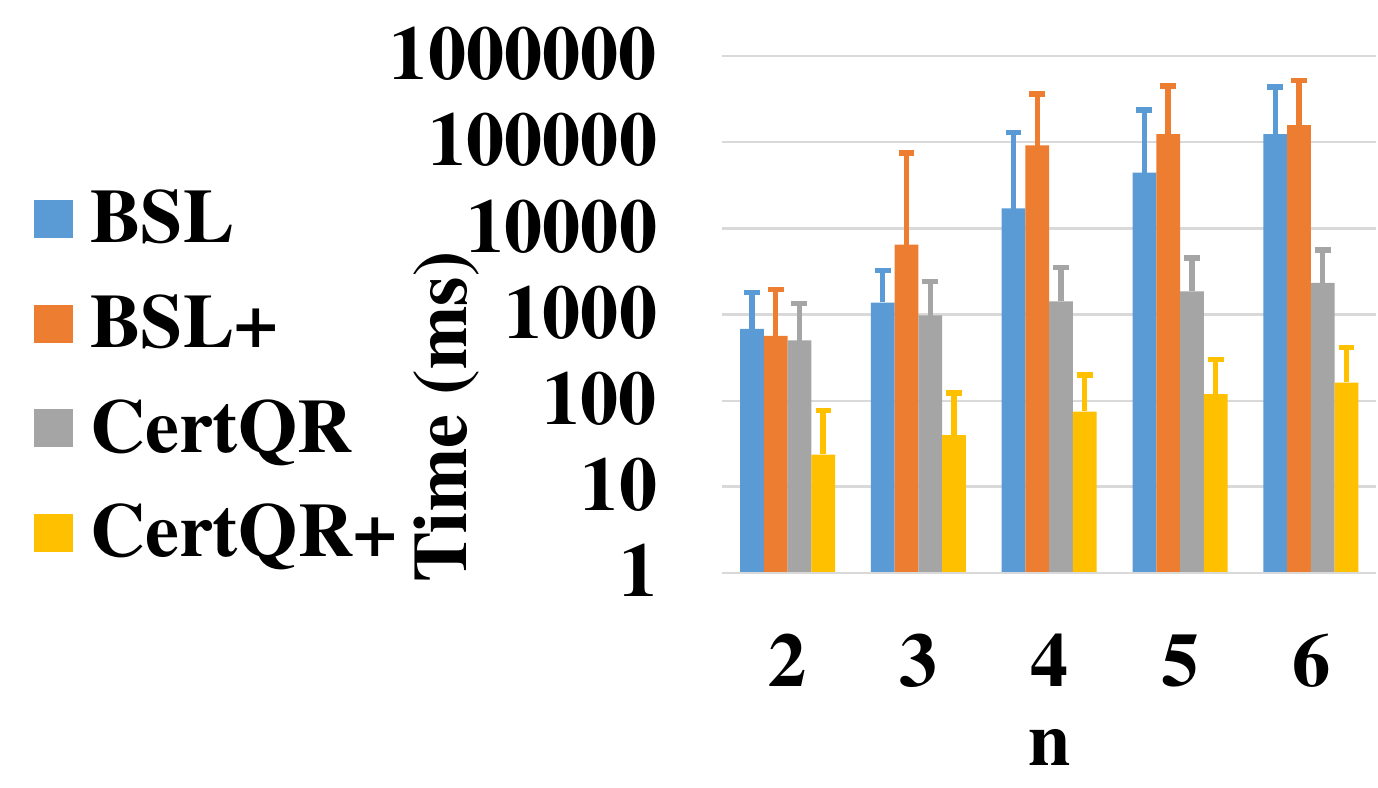}\label{fig:mdban}}
\hfill
\subfloat[Mondial (simulated queries)]{\includegraphics[width=0.6\columnwidth]{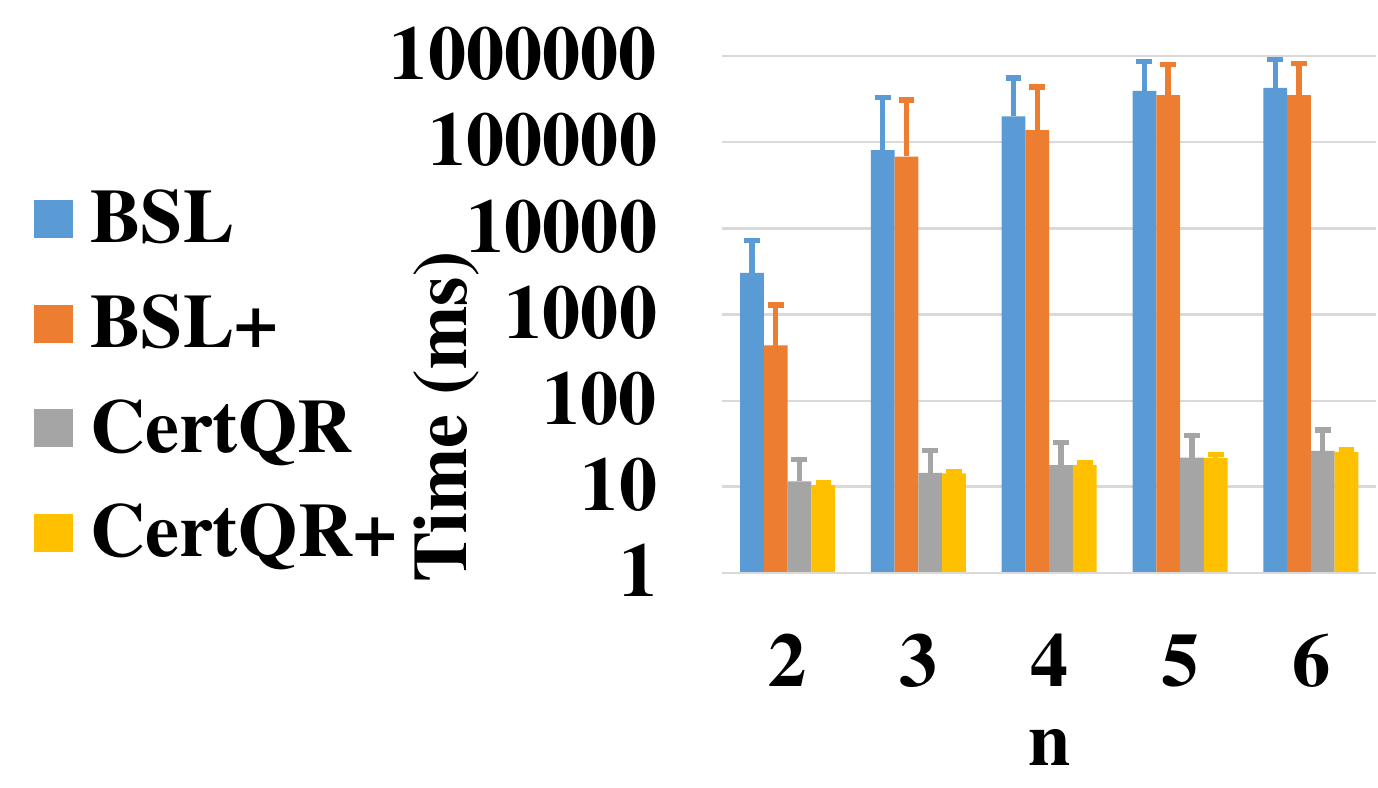}\label{fig:mondan}} \\
\subfloat[DBpedia (random queries)]{\includegraphics[width=0.6\columnwidth]{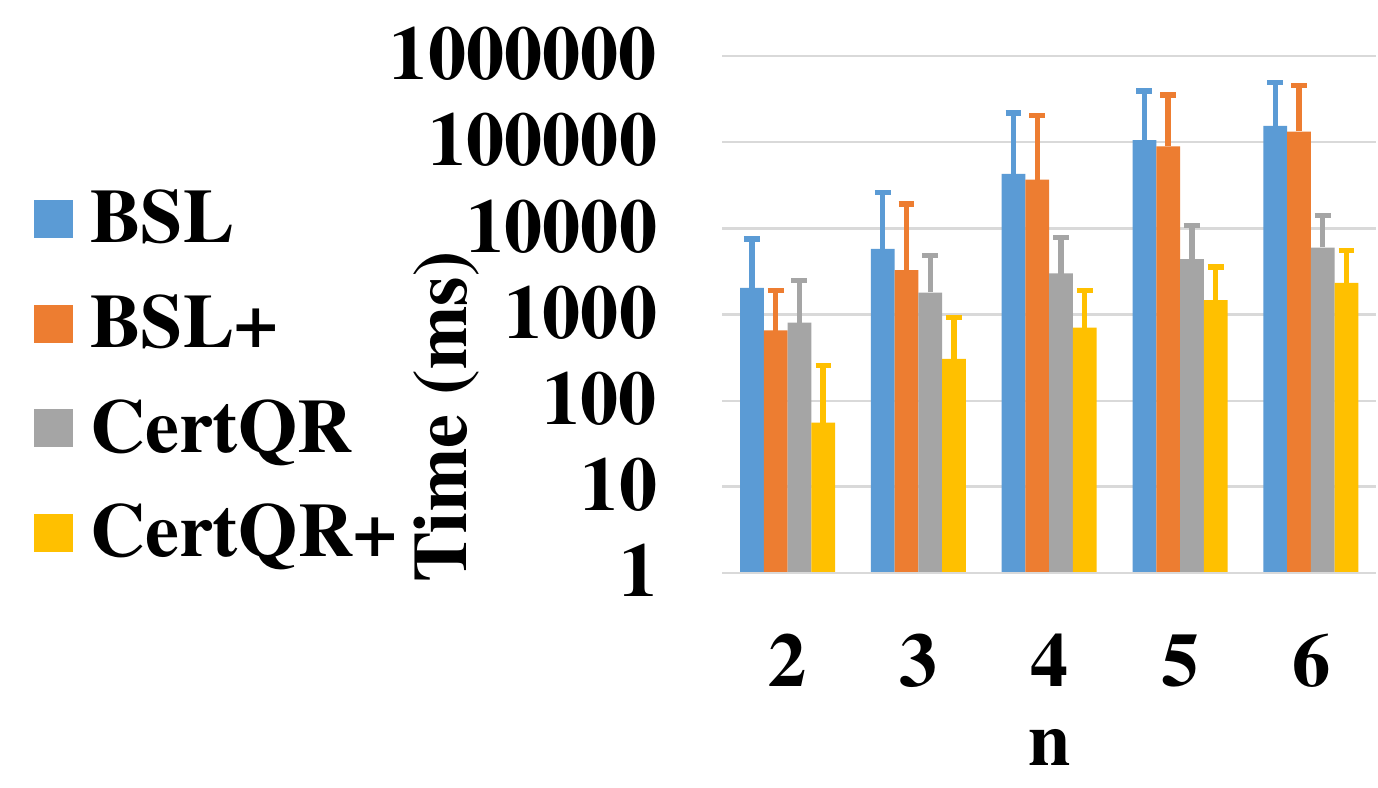}\label{fig:dbprn}}
\hfill
\subfloat[LinkedMDB (random queries)]{\includegraphics[width=0.6\columnwidth]{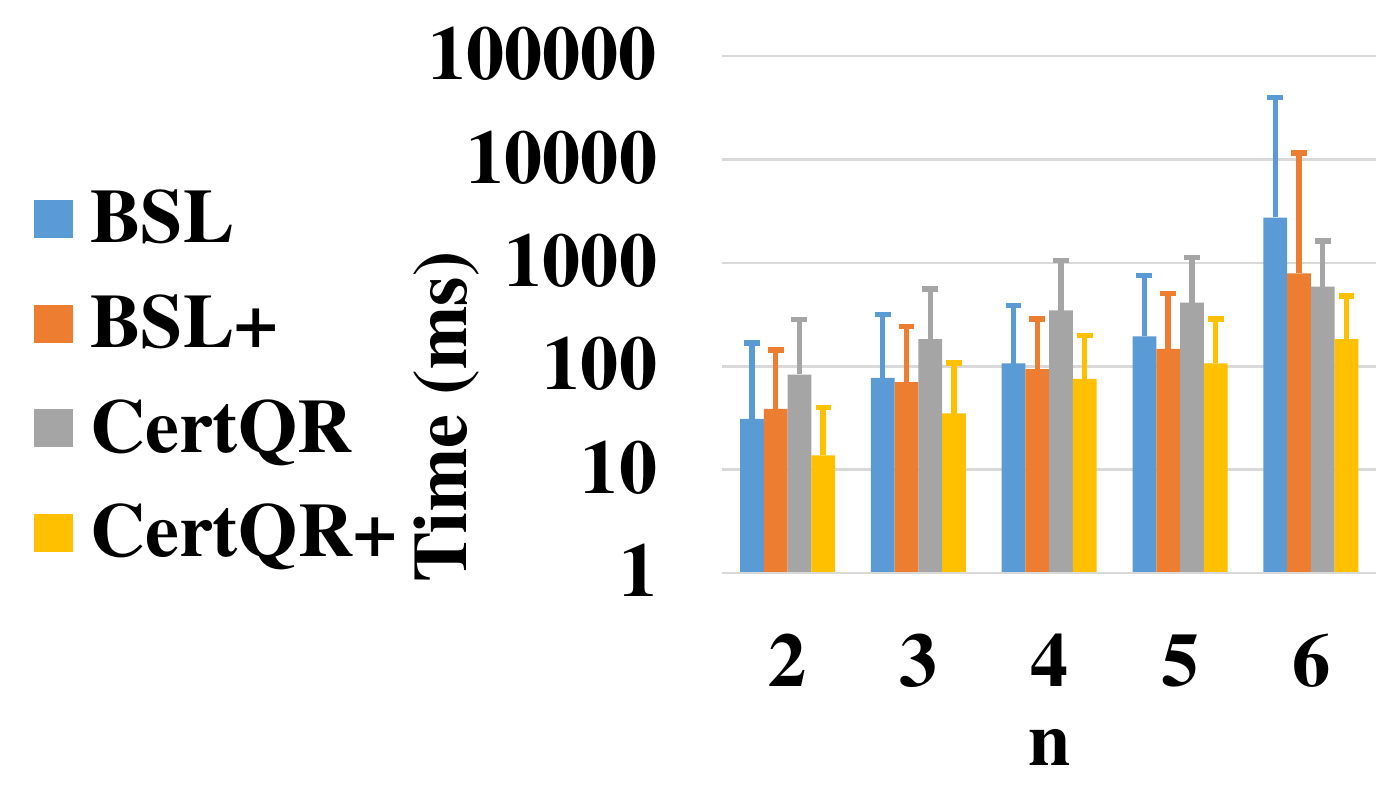}\label{fig:mdbrn}}
\hfill
\subfloat[Mondial (random queries)]{\includegraphics[width=0.6\columnwidth]{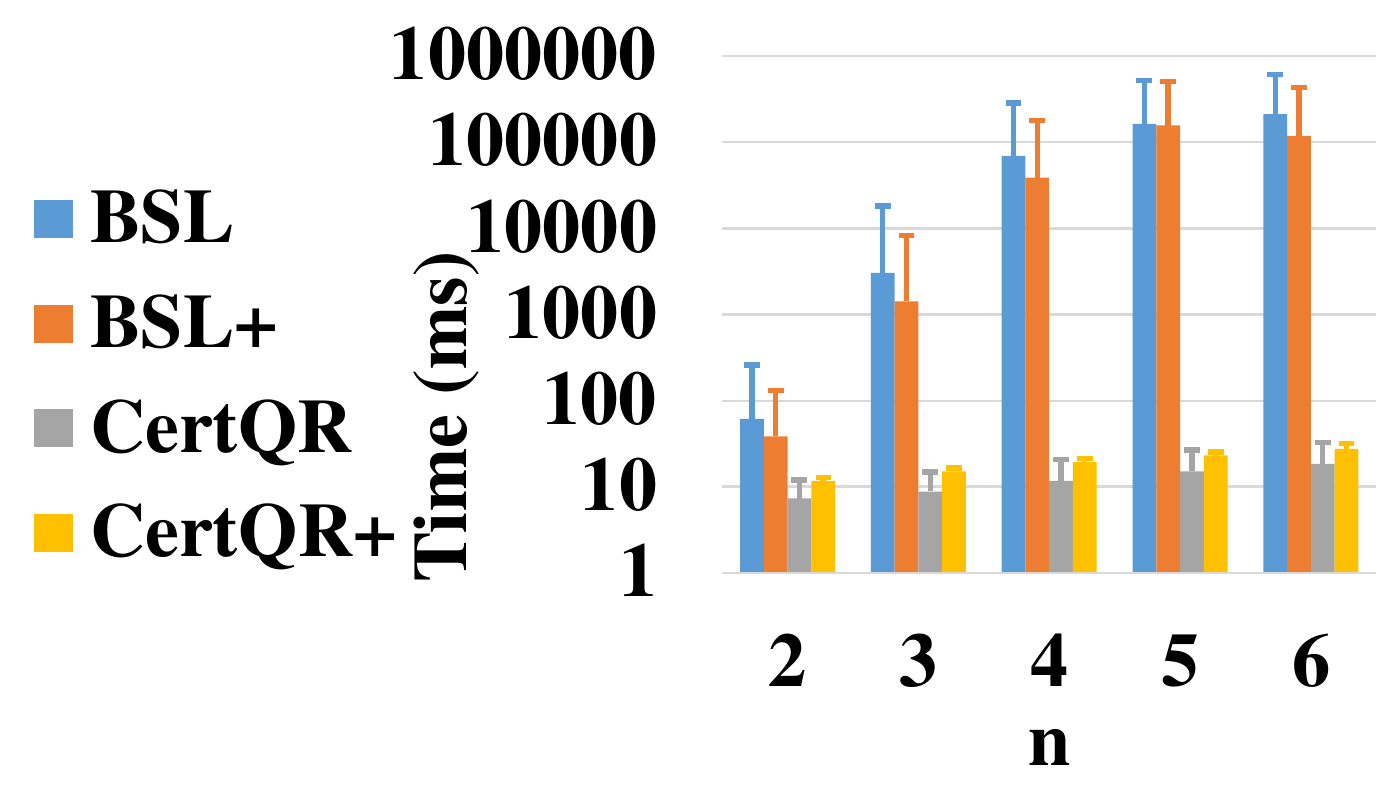}\label{fig:mondrn}}
\caption{Time per run, broken down by number of query entities (i.e.,~$n$).}
\label{fig:n}
\end{figure*}

\begin{figure*}[!t]
\centering
\subfloat[DBpedia (simulated queries)]{\includegraphics[width=0.6\columnwidth]{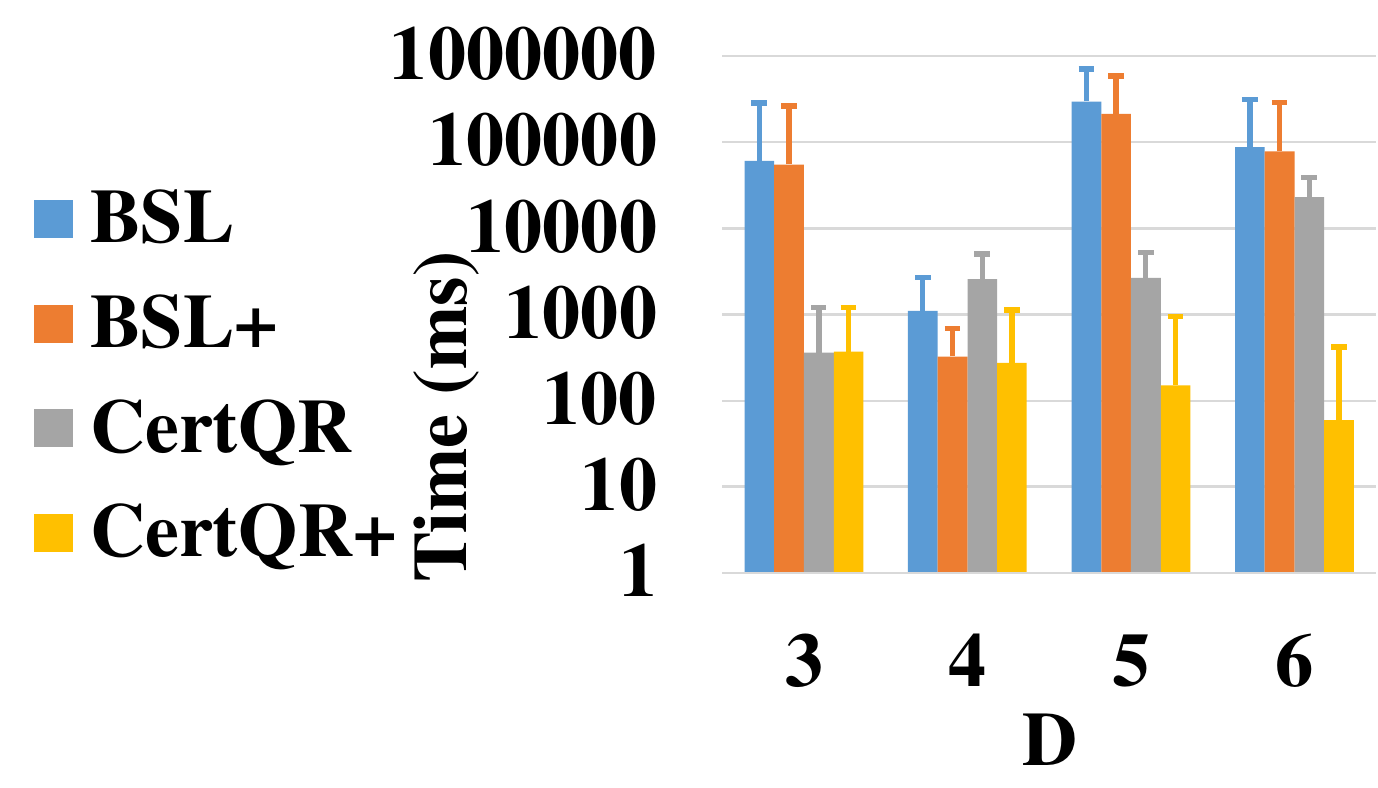}\label{fig:dbpad}}
\hfill
\subfloat[LinkedMDB (simulated queries)]{\includegraphics[width=0.6\columnwidth]{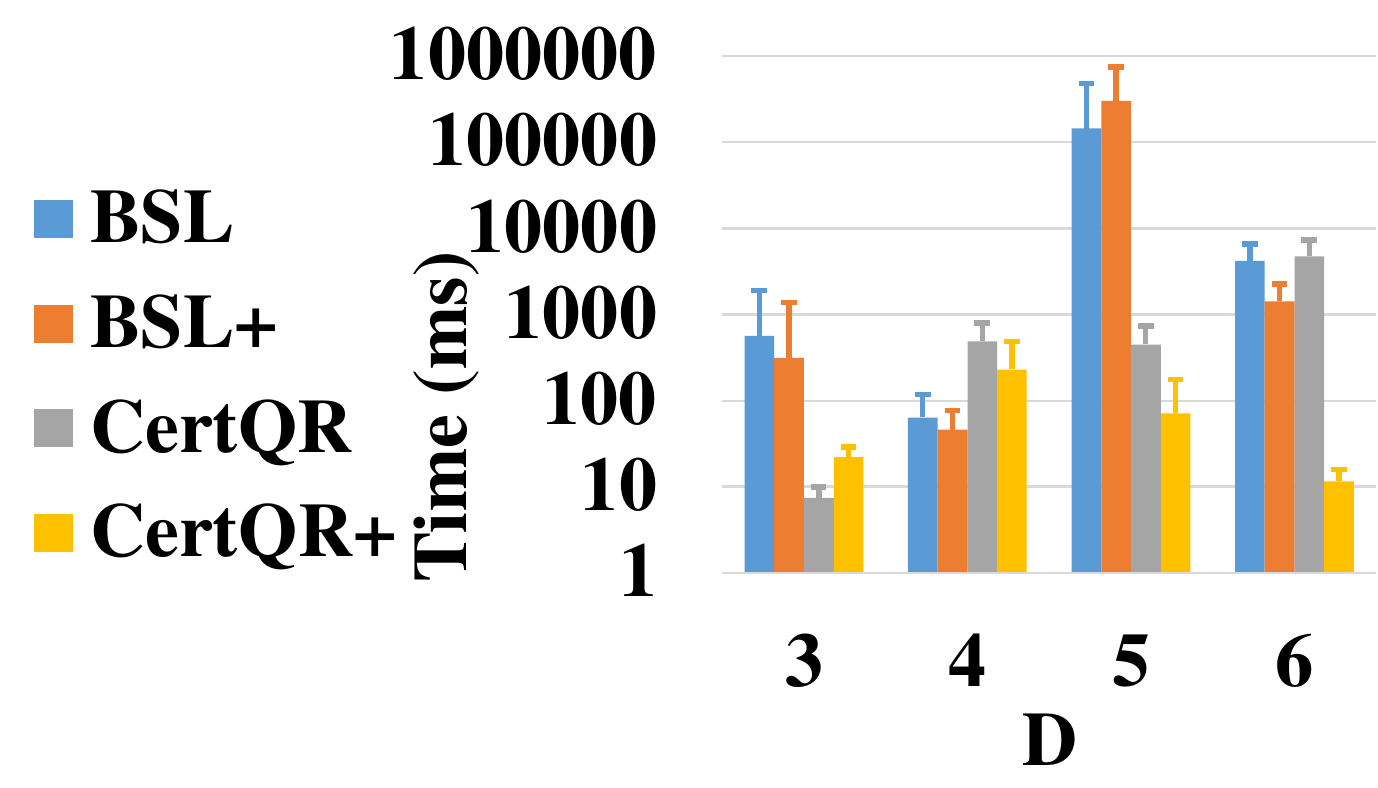}\label{fig:mdbad}}
\hfill
\subfloat[Mondial (simulated queries)]{\includegraphics[width=0.6\columnwidth]{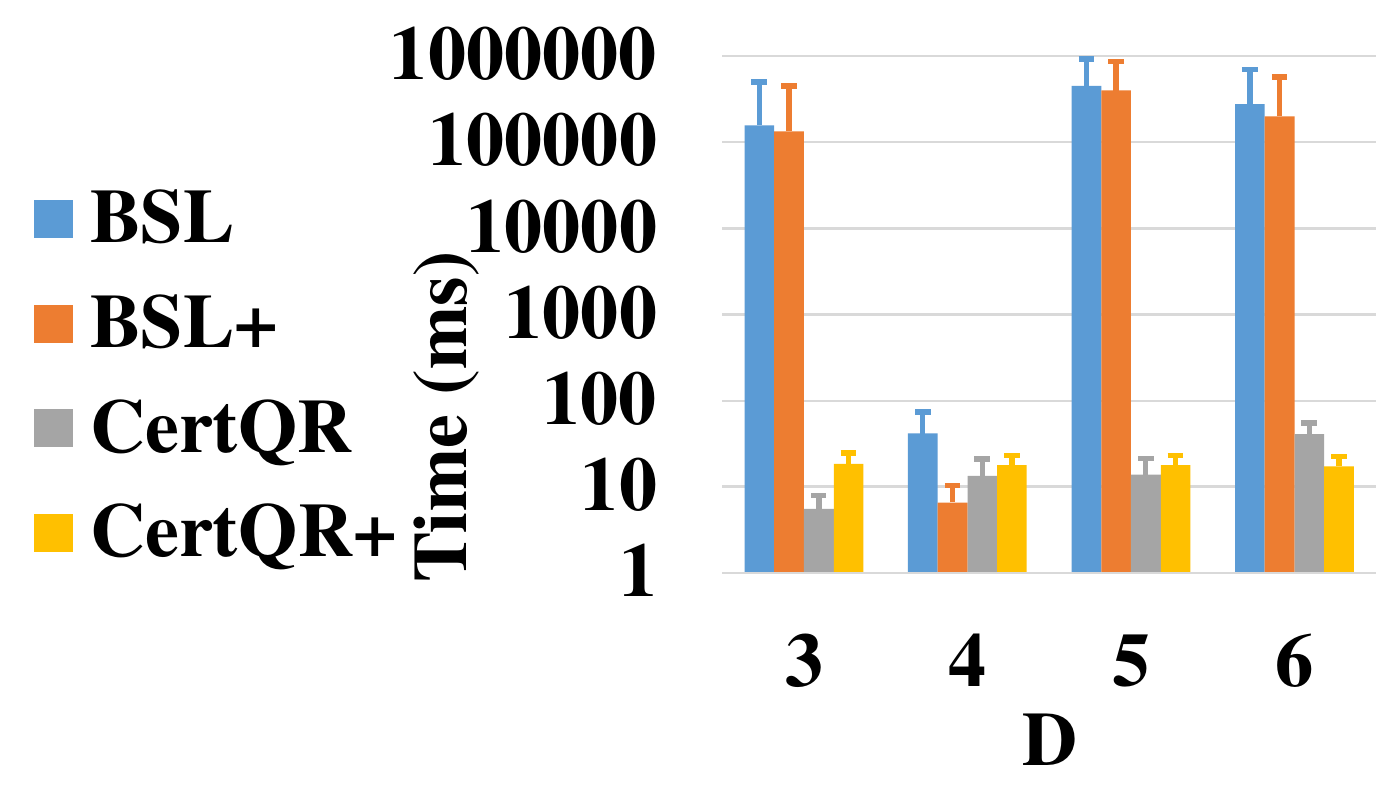}\label{fig:mondad}} \\
\subfloat[DBpedia (random queries)]{\includegraphics[width=0.6\columnwidth]{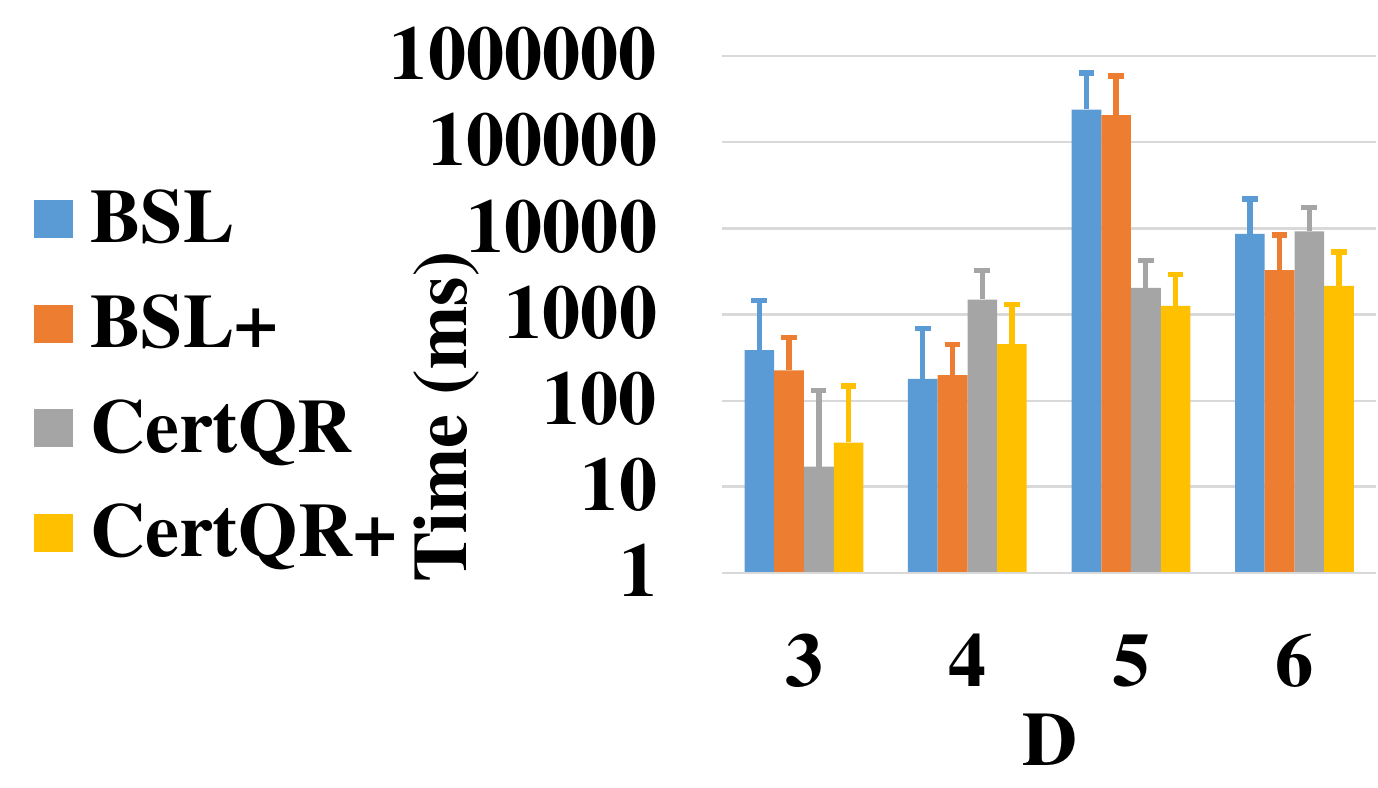}\label{fig:dbprd}}
\hfill
\subfloat[LinkedMDB (random queries)]{\includegraphics[width=0.6\columnwidth]{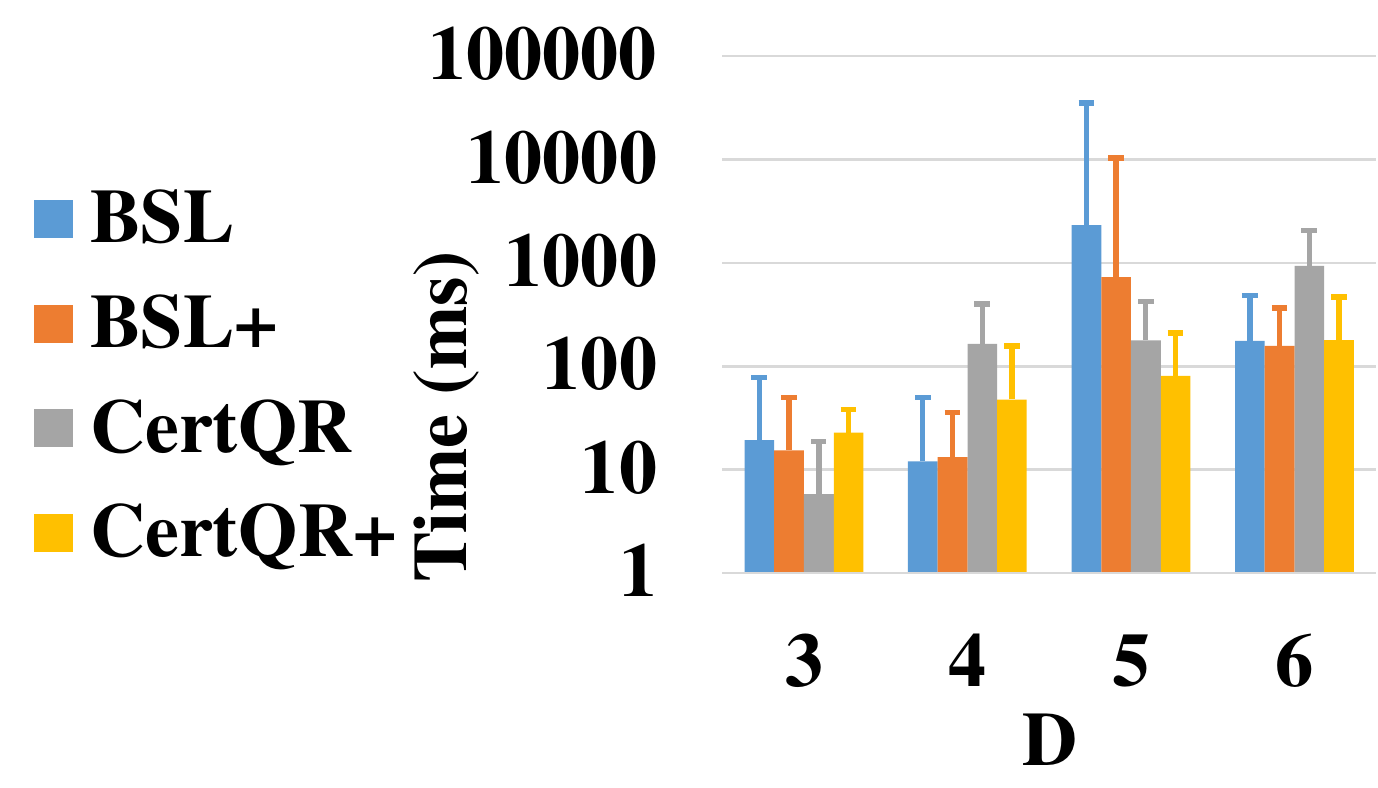}\label{fig:mdbrd}}
\hfill
\subfloat[Mondial (random queries)]{\includegraphics[width=0.6\columnwidth]{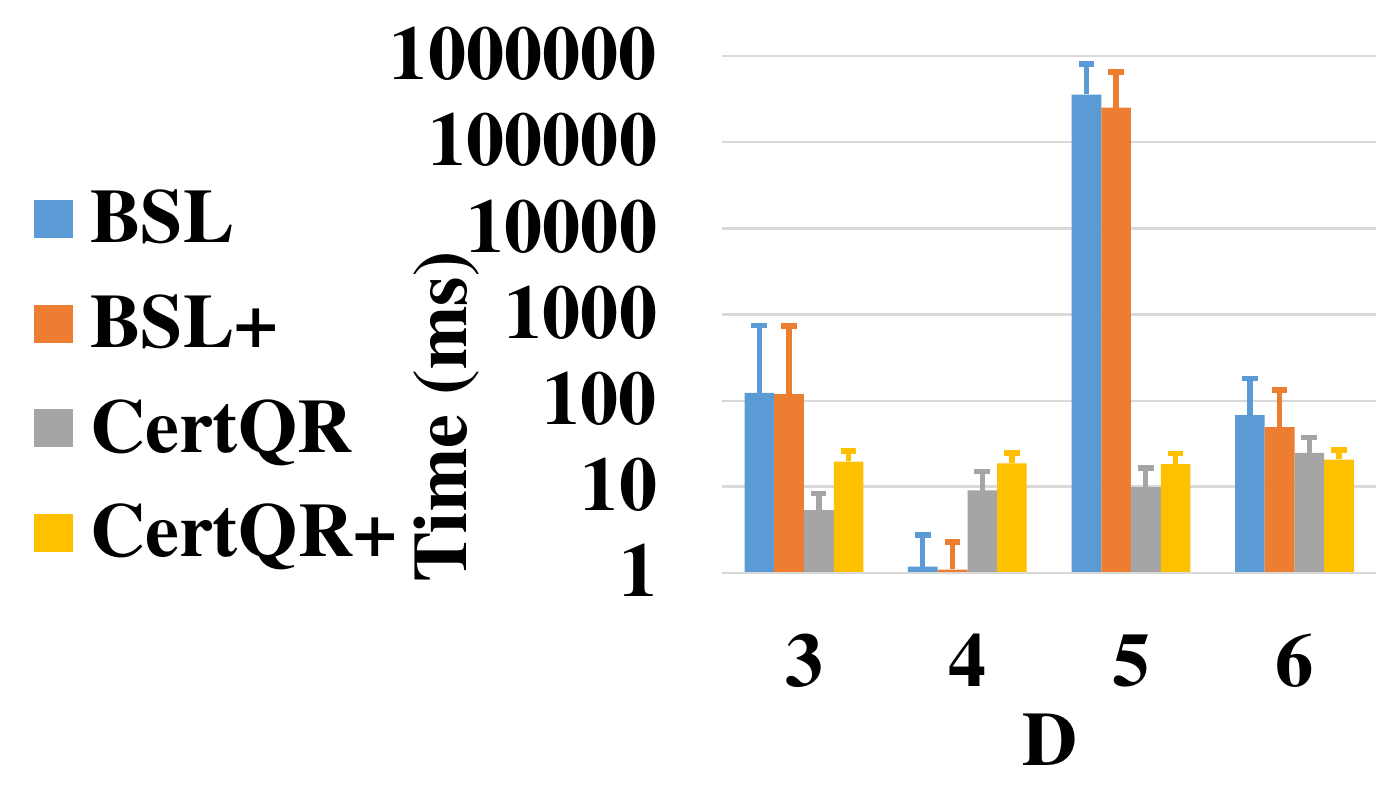}\label{fig:mondrd}}
\caption{Time per run, broken down by diameter constraint (i.e.,~$D$).}
\label{fig:d}
\end{figure*}

\subsubsection{Experiment Results and Analysis: CertQR(+) versus BSL(+)}
Table~\ref{tab:tobsl} and Table~\ref{tab:tobslp} show the proportions of runs where BSL and BSL+ reached timeout, respectively.
Timeout was more often when the number of query entities (i.e.,~$n$) increased,
and became notable from $n=4$,
suggesting the difficulty of the problem and the unscalability of this intuitive solution.
By comparison, our proposed CertQR and CertQR+ never reached timeout in the experiment,
showing their scalability.

The running time of BSL, BSL+, CertQR, and CertQR+ are shown in Fig.~\ref{fig:n} and Fig.~\ref{fig:d} on a logarithmic scale, with lines representing standard deviations.
In Fig.~\ref{fig:n} where the results were broken down by number of query entities (i.e.,~$n$),
although the improved BSL+ slightly outperformed BSL,
both of them were notably slower than CertQR on DBpedia (Fig.~\ref{fig:dbpan} and Fig.~\ref{fig:dbprn}) and Mondial (Fig.~\ref{fig:mondan} and Fig.~\ref{fig:mondrn}).
The differences were larger than an order of magnitude in most cases,
and became larger when $n$~increased,
which \emph{demonstrated the effectiveness of using certificates
and showed the higher efficiency of CertQR}.

CertQR+ consistently outperformed CertQR.
Their differences on large graphs,
i.e., DBpedia (Fig.~\ref{fig:dbpan} and Fig.~\ref{fig:dbprn}) and LinkedMDB (Fig.~\ref{fig:mdban} and Fig.~\ref{fig:mdbrn}),
were about an order of magnitude,
which \emph{demonstrated the effectiveness of our distance-based estimation for best-first search}.
In particular, CertQR+ used less than 1~second under most settings, and used 2~seconds only on occasion (in Fig.~\ref{fig:dbprn}).
Its performance was sufficient for practical use.

One exception was on random queries on LinkedMDB (Fig.~\ref{fig:mdbrn}),
where BSL and BSL+ outperformed CertQR though still were second to CertQR+.
The fairly good performance of BSL and BSL+ was mainly attributed to the sparseness of LinkedMDB.
According to Table~\ref{tab:ds}, the arc-vertex ratio of LinkedMDB was much lower than that of DBpedia and Mondial.
Recall that BSL and BSL+ used a path-merging based algorithm to search for an SA.
For a random query which usually consisted of distantly connected entities,
paths starting from them rarely met when the graph was sparse,
so fewer merging operations were performed.

\begin{table}[!t]
\caption{Performance Improvement on CertQR+}
\label{tab:heu}
\centering
\begin{tabular}{llrrr}
    \hline
    Dataset & Query & dg & ds & dgs \\
    \hline
    DBpedia & Simulated & 19.55\% & 13.52\% & \textbf{23.43\%} \\
    & Random & 15.41\% & 6.07\% & \textbf{18.99\%} \\
    LinkedMDB & Simulated & 26.97\% & 15.12\% & \textbf{31.20\%} \\
    & Random & 17.94\% & 9.72\% & \textbf{17.96\%} \\
    Mondial & Simulated & -4.08\% & -1.74\% & -0.82\% \\
    & Random & 1.93\% & 1.04\% & \textbf{5.23\%} \\
    \hline
\end{tabular}
\end{table}

In Fig.~\ref{fig:d} where the results were broken down by diameter constraint (i.e.,~$D$),
CertQR+ was also generally the fastest.
We would like to clarify two phenomena that might confuse.
First, the performance of BSL and BSL+ fluctuated,
because they ran slower when $D$~was odd.
In that case, merging a set of paths of length~$\left\lceil{\frac{D}{2}}\right\rceil$
could form an SA of diameter $D+1$ and violate the diameter constraint,
so the algorithm had to continue with other merging options.
Second, for simulated queries (Fig.~\ref{fig:dbpad}, Fig.~\ref{fig:mdbad}, and Fig.~\ref{fig:mondad}) which were successful in most cases, CertQR+ generally used less time when $D$~increased,
because the number of certificates also increased.
It became easier for best-first search to find a certificate,
whereas the search space of other algorithms grew,
which \emph{showed the scalability of CertQR+ from another perspective}.

\subsubsection{Experiment Results and Analysis: dg, ds, and dgs}
Table~\ref{tab:heu} shows the performance improvement by adding the three heuristics to CertQR+.
In most cases, both dg and ds considerably improved the performance.
Their combination (i.e., dgs) produced the best results.
The results suggested that \emph{the two fine-grained heuristics were both effective,
and their effects were complementary.}
However, adding heuristics caused negative effects on Mondial,
where most simulated queries were successful when $D$~was small according to Fig.~\ref{fig:mondarn}.
Because the graph was small and dense,
there were many certificates which could be quickly found by CertQR+ without using any heuristics.
Only in that case, adding heuristics did not help much, but their computation took additional time.
\section{Related Work}\label{sect:rw}
SA search has attracted vast research attention from the Semantic Web and database communities
\cite{semrank,rex,rho,explass,vldb04,brahms,conkar,sisp,iswc16,tkde17,ming,star,ceps,explain,relfinder}.
Whereas existing efforts are focused on search algorithms and ranking criteria,
\emph{we present the first study on query relaxation},
which is useful when the compactness of allowable SAs is constrained~\cite{explass,brahms,sisp,rex,ming,ceps,iswc16,tkde17,explain}.
This new research problem and the techniques we use are fundamentally different from those considered in the literature.
In particular, to verify the success of a relationship query,
\emph{we search for a certificate entity
instead of expensively searching for an SA}.
We mainly exploit distances between entities
to prune the search space and improve the performance.
Distance is also used in related research~\cite{iswc16,blinks,banks2,rclique,pruneddp}.
By comparison, our theoretical contributions are distinguished by:
\emph{a distance-based certificate based on which the success of a query can be more efficiently verified (i.e., Theorem~\ref{THE:CERT})},
and \emph{a distance-based estimation which guarantees the optimality of best-first search (i.e., Theorem~\ref{THE:OPT})}.

Query relaxation for other related tasks have been formulated in completely different ways.
In entity search, a query consists of a set of property values describing entity targets.
Query relaxation allows an answer to have property values
that are not exactly the same as but similar to those specified in the query~\cite{qrentity}.
For a path query formulated using regular expressions characterizing path targets,
it can be relaxed to less specific expressions based on inference rules~\cite{qrpath}.
More general graph queries (e.g., SPARQL queries for RDF) can be relaxed
by substituting constants with variables,
and by removing its constituents (e.g., triple patterns) or making them optional~\cite{qrsparql1,qrsparql2}.
Clearly, these ad hoc solutions could not be directly applied to the problem we consider.
A relationship query consists of a set of entities, which is different from the above queries.
\section{Conclusion and Future Work}\label{sect:concl}
The ability to relax a failing relationship query and provide alternative results
improves the usability of an SA search system.
We show that simply relaxing the compactness constraint is impracticable,
and we turn to minimally relaxing the query entities.
Our proposed certificate-based best-first search algorithm is more scalable than baselines,
and its performance could meet the demands of practical use.
We believe its application is not restricted to SA search.
For example, our proposed algorithm can be straightforwardly extended to relax keyword queries on graphs
where each query keyword can be mapped to multiple query entities.
This will be our future work.

Our solution has the following limitations.
First, our algorithm depends on fast distance calculation.
In this work we implement a distance oracle to achieve a trade-off between time and space.
However, it requires rebuild when an entity-relation graph evolves.
Besides, resource-limited machines may not store it in memory, which could notably influence the performance.
We will seek better substitute techniques in future work.
Second, although our algorithm has minimized result incompleteness by finding a maximum successful sub-query,
a user may still be unsatisfied with relaxed results due to missing query entities.
One potentially better solution is to perform replacement rather than removal,
e.g., to replace some query entities with other similar entities.
It would be interesting to conduct a user study to compare different kinds of solutions.
Third, we separate query relaxation from SA search in order to generalize our solution to a wider range of application.
However, it may be more efficient to have a hybrid algorithm that directly outputs relaxed top-ranked SAs.
We will explore this direction based on some common ranking criteria.

\section*{Acknowledgment}
This work was supported
in part by the NSFC under Grant 61772264, and in part by the Six Talent Peaks Program of Jiangsu Province under Grant RJFW-011.

\bibliography{main}

\end{document}